%% file: main.tex
\documentclass[11pt]{article} %
\pdfoutput=1
\usepackage[margin=1in]{geometry}
\usepackage[font={small,bf}]{caption}
\usepackage[colorlinks,linkcolor=ForestGreen,citecolor=ForestGreen,
backref, bookmarks, bookmarksopen, bookmarksnumbered]{hyperref}
\usepackage{amsmath,amssymb, amsthm}    %
\usepackage{thmtools}
\usepackage{thm-restate}
\usepackage{bm}
\usepackage{bbm}
\usepackage[algo2e]{algorithm2e} %
\usepackage{algorithm,algpseudocode,algorithmicx}
\usepackage{enumerate}

\newtheorem{theorem}{Theorem}
\newtheorem{definition}{Definition}
\newtheorem{lemma}{Lemma}

\newtheorem{corollary}{Corollary}

\usepackage{verbatim}           %
\usepackage{graphicx,float}     %
\usepackage[small,compact]{titlesec}
\usepackage{setspace}

\usepackage{color}
\definecolor{ForestGreen}{rgb}{0.1333,0.5451,0.1333}
\usepackage{url}
\usepackage{enumitem}
\usepackage{xfrac} 
\usepackage[toc, page]{appendix}
\usepackage{physics}
\usepackage{psfrag}
\usepackage{xifthen} 
\numberwithin{equation}{section}
\usepackage{booktabs, array, makecell}

\input defs.tex

\title{
Computing Lewis Weights to High Precision 
} 
\date{}
\author{
			Maryam Fazel\thanks{\texttt{mfazel@uw.edu}. University of Washington.} 
			\and
			Yin Tat Lee\thanks{\texttt{yintat@uw.edu}. University of Washington.}
			\and
			Swati Padmanabhan\thanks{\texttt{pswati@uw.edu}. University of Washington.}
			\and
		    Aaron Sidford\thanks{\texttt{sidford@stanford.edu}. Stanford University}
}

\begin{document}

\begin{titlepage}
\maketitle
\begin{abstract}%
We present an algorithm for computing approximate $\ell_p$ Lewis weights to high precision. Given a full-rank $\mathbf{A} \in \mathbb{R}^{m \times n}$ with $m \geq n$ and a scalar $p>2$, our algorithm computes $\epsilon$-approximate $\ell_p$ Lewis weights of $\mathbf{A}$ in  $\widetilde{O}_p(\log(1/\epsilon))$ iterations; the cost of each iteration is linear in the input size plus the cost of computing the leverage scores of  $\mathbf{D}\mA$ for diagonal $\mathbf{D} \in \mathbb{R}^{m \times m}$. Prior to our work, such a computational complexity was known only for $p \in (0, 4)$ \cite{CohenPeng2015}, and combined with this result, our work yields the first polylogarithmic-depth polynomial-work algorithm for the problem of computing $\ell_p$ Lewis weights to high precision for \emph{all} constant $p > 0$. An important consequence of this result is also the first polylogarithmic-depth polynomial-work algorithm for computing a nearly optimal self-concordant barrier for a polytope.

\end{abstract}
\thispagestyle{empty}
\end{titlepage}

\newpage

\input{introduction.tex}
\input{our_analysis.tex}
\newpage
\bibliography{main.bib}
\newpage
\input{appendices.tex}

\end{document}

%% file: introduction.tex
\section{Introduction to Lewis Weights}\label{sec:intro}

In this paper, we study the problem of computing the $\ell_p$ Lewis weights\footnote{From hereon, we refer to these simply as ``Lewis weights'' for brevity.} of a matrix. 
\begin{definition}\cite{lewis1978finite, CohenPeng2015}\label{def-lewisweights}
Given a full-rank matrix $\mA\in \R^{m\times n}$ with $m \geq n$ and a scalar  $p \in (0, \infty)$, the  Lewis weights of $\mA$ are the entries of the unique\footnote{Existence and uniqueness was first proven by D.R.Lewis \cite{lewis1978finite}, after whom the weights are named.} vector $\lw\in \R^m$ satisfying the equation 
\[\lw_i^{2/p} = a_i^\top (\mA^\top \overline{\mW}^{1-2/p} \mA)^{-1} a_i \text{  for all $i\in [m]$}, \numberthis\label{eq_def_lewis_weights}\] where $a_i$ is the $i$'th row of matrix $\mA$ and $\overline{\mW}$ is the diagonal matrix with vector $\lw$ on the diagonal.
\end{definition}

\textbf{Motivation.} We contextualize our problem with a simpler, geometric notion. Given a set of $m$ points $\{a_i\}_{i=1}^m \in \R^n$ (the rows of the preceding matrix $\mA\in \R^{m \times n}$), their \emph{John ellipsoid} \cite{john1948extremum}  is the minimum\footnote{The John ellipsoid may also refer to the maximal volume ellipsoid enclosed by the set $\{x: |x^\top a_i| \leq 1\}$, but in this paper, we use the former definition.} volume ellipsoid enclosing them. This ellipsoid finds use  across experiment design and computational geometry \cite{Todd2016} and is central to certain cutting-plane methods \cite{vaidya1989new, lee2015faster}, an algorithm fundamental to mathematical optimization (Section~\ref{sec:appsRelatedWork}). It turns out that the John ellipsoid of a set of points $\{a_i\}_{i=1}^m \in \R^n$ is expressible \cite{BVbook} as the solution to the following convex program, with the objective being a stand-in for the volume of the ellipsoid and the constraints encoding the requirement that each given point $a_i$ lie \emph{within} the ellipsoid: 
\[ \mbox{minimize}_{\mM \succeq 0} \det(\mM)^{-1}, \textup{ subject to } a_i^\top \mM a_i \leq 1, \textup{ for all } i \in [m]. \numberthis\label{opt-johnell}
\] The problem \eqref{opt-johnell} may be generalized by the following convex program  \cite{wojtaszczyk1996banach, CohenPeng2015}, the generalization immediate from substituting $p = \infty$ in \eqref{opt-lewisell}:
\[ \mbox{minimize}_{\mM \succeq 0} \det(\mM)^{-1}, \textup{ subject to } \sum_{i = 1}^m (a_i^\top \mM a_i)^{p/2} \leq 1. \numberthis\label{opt-lewisell}
\] Geometrically, \eqref{opt-lewisell} seeks the minimum volume ellipsoid with a bound on the $p/2$-norm of the distance of the points to the ellipsoid, and its solution $\mM$  is  the ``Lewis ellipsoid'' \cite{CohenPeng2015} of $\{a_i\}_{i = 1}^m$. 

The optimality condition of \eqref{opt-lewisell}, written using $\lw\in \R^m$ defined as $\lw_i \defeq (a_i^\top \mM a_i)^{p/2}$, is equivalent to \eqref{eq_def_lewis_weights},  and this demonstrates that solving \eqref{opt-lewisell} is one approach to obtaining the Lewis weights of $\mA$ (see \cite{CohenPeng2015}). This equivalence also underscores the fact that the problem of computing Lewis weights is a natural $\ell_p$ generalization of the problem of computing the John ellipsoid.

More broadly,   Lewis weights are ubiquitous across statistics, machine learning, and mathematical optimization in diverse applications, of which we presently highlight two (see Section~\ref{sec:appsRelatedWork} for details). First, their interpretation as ``importance scores'' of rows of matrices makes them key to shrinking the row dimension of input data \cite{drineas2006sampling}.  Second, through their role in constructing self-concordant barriers of polytopes \cite{LeeSidford2014}, variants of Lewis weights have found prominence in recent advances in the computational complexity of linear programming. 

From a purely optimization perspective, Lewis weights may be viewed as the optimal solution to the following convex optimization problem (which is in fact essentially dual to \eqref{opt-lewisell}): 
\[
\lw = \arg\min_{w\in \mathbb{R}^m_{> 0}} \fobj(w)\defeq -\log \det (\mA^\top \mW \mA) + \frac{1}{1+\alpha}\ones^\top w^{1+\alpha}, \textrm{ for $\alpha = \tfrac{2}{p-2}$}.\numberthis\label{def_objective}
\] As elaborated in  \cite{CohenPeng2015, DBLP:journals/corr/abs-1910-08033}, the reason this problem yields the Lewis weights is that
 an  appropriate scaling of  its solution $\lw$ transforms its optimality condition from $\lw_i^{\alpha} = a_i^\top (\mA^\top \overline{\mW} \mA)^{-1} a_i$ to \eqref{eq_def_lewis_weights}. The problem \eqref{def_objective} is  a simple and natural one and, in the case of $\alpha = 1$ (corresponding to the John ellipsoid), has been the subject of study for designing new optimization methods \cite{Todd2016}. 

In summary, Lewis weights naturally arise as generalizations of extensively studied problems in convex geometry and optimization. This, coupled with their role in machine learning, makes understanding the complexity of computing Lewis weights, i.e., solving \eqref{def_objective}, a fundamental problem. 

\paragraph{Our Goal.}We aim to design \emph{high-precision} algorithms for computing \emph{$\eps$-approximate}  Lewis weights, i.e., a vector $w\in \R^m$ satisfying 
\[ w_i \approx_{\eps} \lw_i, \text{  for all $i\in [m]$}, \text{ where } \lw \text{ is defined in } \eqref{eq_def_lewis_weights} \text{ and } \eqref{def_objective}. \numberthis\label{def-lewiswts-obj-approxLW}\] where $a \approx_\eps b$ is used to denote $(1-\eps) a \leq b \leq (1+\eps) a$. To this end, we design algorithms to solve the convex program \eqref{def_objective} to $\teps$-additive accuracy for an appropriate $\teps = \textrm{poly}(\eps, n)$, which we prove suffices in Lemma~\ref{lem-fn-to-LW}. 

By a ``high-precision'' algorithm, we mean one with a runtime polylogarithmic in $\eps$. We emphasize that for several applications such as randomized sampling \cite{CohenPeng2015},  \textit{approximate} Lewis weights suffice; however, we believe that  high-precision methods such as ours enrich our understanding of the structure of the optimization problem \eqref{def_objective}. Further, as stated in Theorem ~\ref{selfconcordancethm}, such  methods yield new runtimes for directly computing a near-optimal self-concordant barrier for polytopes.

We use \emph{number of leverage score computations}  as the complexity measure of our algorithms. Our choice is a result of the fact that leverage scores of appropriately scaled matrices appear in both $\nabla\fobj(w)$ (see Lemma~\ref{lem_gradAndHess}) and in the verification of correctness of Lewis weights. This measure of complexity stresses the \emph{number of iterations} rather than the details of iteration costs (which depend on exact techniques used for leverage core computation, e.g.,  fast matrix multiplication) and is consistent with many prior algorithms (see Table~\ref{tableresults}).

\paragraph{Prior Results.} The first polynomial-time algorithm for computing Lewis weights  was presented by \cite{CohenPeng2015} and performed only $\widetilde{O}_p(\log(1/\eps))\footnote{We use $O_p$ to hide a polynomial in $p$ and $\widetilde{O}$ and $\widetilde{\Omega}$ to hide factors polylogarithmic in $p, n$, and $m$.}$ leverage score computations.  However, \emph{their result holds only for $p\in (0, 4)$}. We explain the source of this limited range in Section~\ref{sec:overviewofapproach}.

In comparison, for $p\geq 4$, existing algorithms are slower: the algorithms by \cite{CohenPeng2015}, \cite{YinTatThesis}, and \cite{DBLP:journals/corr/abs-1910-08033} perform $\widetilde{\Omega}(n)$, $\widetilde{O}(1/\eps)$, and $\widetilde{O}(\sqrt{n})$ leverage score computations, respectively. \cite{CohenPeng2015} also gave an algorithm with total runtime $\Ord(\tfrac{1}{\eps}\nnz{\mA} + c_p n^{O(p)})$. Of note is the fact that the algorithms with  runtimes polynomial in $1/\eps$ (\cite{YinTatThesis,CohenPeng2015}) satisfy the weaker approximation condition 
$\lw_i^{2/p}  \approx_{\eps} a_i^\top (\mA^\top \overline{\mW}^{1-2/p} \mA)^{-1} a_i$,  which is in fact implied by our condition \eqref{def-lewiswts-obj-approxLW}. 

We display these runtimes in Table~\ref{tableresults}, assuming that the cost of a leverage score computation is $O(mn^2)$ (which, we reiterate, may be reduced through the use of fast matrix multiplication).  In terms of the number of leverage score computations, Table~\ref{tableresults} highlights the contrast between the \emph{polylogarithmic} dependence on input size and accuracy for $p \in (0, 4)$ and \emph{polynomial} dependence on these factors for $p \geq 4$.  \emph{The motivation behind our paper is to close this gap.} 
\subsection{Our Contribution} 

We design \emph{an algorithm that computes Lewis weights to high precision for all $p>2$ using only $\widetilde{O}_p(\log(1/\eps))$ leverage score computations}. Together with \cite{CohenPeng2015}'s result for $p\in (0, 4)$, our result therefore completes the picture on a near-optimal reduction from leverage scores to Lewis weights for all $p>0$. 

\begin{restatable}[Main Theorem (Parallel)]{theorem}{THMmain}\label{thm_main_thm}
Given a full-rank matrix $\mA \in \R^{m\times n}$ and $p \geq 4$, we can compute (Algorithms~\ref{alg_ourAlg} and ~\ref{alg_roundparallel}) its $\eps$-approximate Lewis weights \eqref{def-lewiswts-obj-approxLW} in $O(p^3 \log(m p/\eps))$ iterations\footnote{Our algorithms work for all $p>2$, as can be seen in our proof in Section~\ref{sec:pf-main-thm-par}. However, for $p \in (2, 4)$, the algorithm of \cite{CohenPeng2015} is faster, and therefore, in our main theorems, we state runtimes only for $p\geq 4$.}. Each iteration computes the leverage scores of a matrix $\mD \mA$ for a diagonal matrix $\mD$. The total runtime is  $O(p^3 m n^2\log(m p/\eps))$, with  $O(p^3 \log(mp/\eps) \log^2 (m))$ depth. 
\end{restatable}

Theorem~\ref{thm_main_thm} is attained by a parallel algorithm for  computing  Lewis weights that consists of polylogarithmic rounds of leverage score computations and therefore has polylogarithmic-depth, a result that was not known prior to this work. 

\begin{restatable}[Main Theorem (Sequential)]{theorem}{THMmainseq}\label{thm_main_thm_seq}
Given a full-rank matrix $\mA \in \R^{m\times n}$ and $p\geq 4$, we can compute (Algorithms~\ref{alg_ourAlg} and ~\ref{alg_roundsequential}) its $\eps$-approximate  Lewis weights \eqref{def-lewiswts-obj-approxLW}  in 
$O(p m \log(m p/\eps))$ iterations. Each iteration computes the leverage score of one row of $\mD \mA$ for a diagonal matrix $\mD$. The total runtime is $O(p m n^2 \log(m p/\eps))$.
\end{restatable}

\begin{remark}\label{remarkJohnEll}
The solution to \eqref{opt-lewisell} characterizes a ``Lewis ellipsoid,'' and the $\ell_{\infty}$ Lewis ellipsoid of  $\mA$ is precisely its John ellipsoid. After symmetrization \cite{Todd2016}, computing the John ellipsoid is equivalent to solving a linear program (LP). Therefore, computing Lewis weights in  $O(\log(mp/\eps))$ iterations would imply a polylogarithmic-depth algorithm for solving LPs, which, given the current  $O(\sqrt{n})$ depth \cite{DBLP:journals/corr/abs-1910-08033}, would be a significant breakthrough in the field of optimization. We therefore believe that it would be difficult  to remove the polynomial dependence on $p$ in our runtime.  
\end{remark}

\begin{table}[ht] 
\centering
\begin{tabular}{lccc}
\toprule 
Authors & Range of $p$ & \thead{Number of \\ Leverage Score\\ Computations/Depth} & Total Runtime \\
\midrule
\cite{CohenPeng2015} & $p \in (0, 4)$ &  $O\left(\tfrac{1}{1-|1-p/2|} \cdot \log(\tfrac{\log (m)}{\eps})\right)$ & $O\left(\tfrac{1}{1-|1-p/2|}\cdot mn^2 \cdot \log(\tfrac{\log (m)}{\eps})\right)$\\
\cite{CohenPeng2015} & $p \geq 4$ & $\Omega(n)$ & $\Omega(mn^3  \cdot \log (\tfrac{m}{\eps}))$\\
\cite{CohenPeng2015}\mbox{*}  & $p \geq 4$ & not applicable & $O\left(\tfrac{\text{nnz}(\mA)}{\eps} + c_p n^{O(p)}\right)$ \\ 
\cite{YinTatThesis}\mbox{*} & $p \geq 4$ & $O\left(\tfrac{1}{\eps} \cdot \log(m/n)\right)$ & $O\left( \left(\tfrac{\text{nnz}(\mA)}{\eps} + \frac{n^3}{\eps^3}\right) \cdot\log(m/n)\right)$\\ 
\cite{DBLP:journals/corr/abs-1910-08033} & $p \geq 4$ & $O(p^2 \cdot n^{1/2} \cdot \log(\tfrac{1}{\eps}))$ & $O(p^2 \cdot mn^{2.5}  \cdot \text{poly}\log(\tfrac{m}{\eps}))$ \\ 
Theorem~\ref{thm_main_thm} & $p \geq 4$ & $O(p^3 \cdot \log (\tfrac{mp}{\eps}))$ & $O(p^3  \cdot mn^2  \cdot \log (\tfrac{mp}{\eps}))$\\ 
\bottomrule
\end{tabular}
\caption{Runtime comparison for computing Lewis weights. 
Results with asterisks use a weaker notion of approximation than our paper \eqref{eq_def_lewis_weights}. All dependencies on $n$ in the running times of these methods can be improved using fast matrix multiplication.
}

\label{tableresults}
\end{table}

\subsection{Overview of Approach}\label{sec:overviewofapproach} 
Before presenting our algorithm, we describe obstacles to directly extending previous work on the problem for $p \in (0,4)$ to the case $p\geq 4$. For $p\in (0,4)$, \cite{CohenPeng2015, DBLP:journals/corr/abs-1910-08033} design algorithms that, with a single computation of leverage scores, make constant (dependent on $p$) multiplicative progress on error (such as function error or distance to optimal point), thus attaining runtimes polylogarithmic in $\eps$. However, these methods crucially rely on \emph{contractive properties} that, in contrast to our work, do \emph{not} necessarily hold for $p\geq 4$. 

For example, one of the algorithms in \cite{CohenPeng2015} starts with a vector $v\approx_{c} \lw$, where $\lw$ is the vector of true Lewis weights and $c$ some constant.  Consequently, we have  $(a_i^\top (\mA^\top \mathbf{V}^{1-2/p} \mA)^{-1} a_i)^{p/2} \approx_{c^{|p/2 -1|}} (a_i^\top (\mA^\top \overline{\mW}^{1-2/p} \mA)^{-1} a_i)^{p/2}$. Due to this map being a contraction  for $|p/2-1|<1$, or equivalently, for  $p\in (0, 4)$,  $O(\log (\log n))$ recursive calls to it give  Lewis weights for $p <4$, but the contraction - and, by extension, this method - does not immediately extend to the setting $p \geq 4$.  

Prior algorithms for $p\geq 4$ therefore resort to alternate optimization techniques. \cite{CohenPeng2015} frames Lewis weights computation as determinant maximization \eqref{opt-lewisell} (see Section~\ref{cohenpeng-determinant}) and applies cutting plane methods \cite{grotschel1981ellipsoid,lee2015faster}.  \cite{YinTatThesis} uses mirror descent,  and \cite{DBLP:journals/corr/abs-1910-08033} uses homotopy methods. These approaches yield runtimes with $\text{poly}(n)$ or $\text{poly}(\tfrac{1}{\eps})$ leverage score computations, and therefore, in order to attain runtimes of $\text{polylog}(1/\eps)$ leverage score computations, we need to rethink the algorithm. 

\paragraph{Our Approach.}  As stated in Section~\ref{sec:intro}, to obtain $\eps$-approximate Lewis weights for $p \geq 4$, we compute a $w$ that satisfies $\fobj(\wopt) \leq \fobj(w)\leq \fobj(\wopt) + \teps$, where $\fobj$ and $\wopt$ are as defined in \eqref{def_objective} and $\teps = O(\mathrm{poly}(n, \eps))$. In light of the preceding bottlenecks in prior work,  we circumvent techniques that directly target constant multiplicative progress (on some potential) in each iteration. 

Our main technical insight is that \textit{when the leverage scores for the current weight $w \in \R^n_{> 0}$ satisfy a certain technical condition (inequality \eqref{ineqgradcondforprog-early})}, it is indeed possible to update $w$ to get multiplicative decrease in function error  ($\fobj(w)-\fobj(\wopt)$), thus resulting in our target runtime. To turn this insight into an algorithm, we design a corrective procedure that ensures that  \eqref{ineqgradcondforprog-early} is always satisfied: in other words, whenever  \eqref{ineqgradcondforprog-early} is \emph{violated}, this procedure updates $w$ so that the new  $w$ \emph{does} satisfy  \eqref{ineqgradcondforprog-early}, setting the stage for the aforementioned multiplicative progress. An important additional property of this procedure is that it does \emph{not} increase the objective function and is therefore in keeping with our goal of minimizing \eqref{def_objective}.

Specifically, the technical condition that our geometric decrease in function error hinges on is 
\[\max_{i\in [m]} \frac{a_i^\top (\mA^\top \mW \mA)^{-1} a_i }{w_i^{\alpha}}\leq 1 + \alpha\,. 
\numberthis
\label{ineqgradcondforprog-early}\] 
This ratio follows naturally from the gradient and Hessian of the function objective (see Lemma~\ref{lem_gradAndHess}). Our algorithm's update rule to solve \eqref{def_objective} is obtained from minimizing a second-order approximation to the objective at the current point, and the condition specified in \eqref{ineqgradcondforprog-early} allows us to relate the progress of a type of quasi-Newton step to lower bounds on the progress there is to make, which is critical to turning a runtime of \polyeps  into  $\text{polylog}(1/\eps)$ (Lemma~\ref{lem_subopt}).

The process of updating $w$ so that \eqref{ineqgradcondforprog-early} goes from being violated to being satisfied corresponds, geometrically, to sufficiently rounding the ellipsoid $\mathcal{E}(w) = \{x: x^\top \mA^\top \mW \mA x\leq 1\}$; specifically, the updated ellipsoid satisfies $\mathcal{E}(w) \subseteq \{\|\mW^{\frac{1}{2-p}} \mA x\|_{\infty} \leq \sqrt{1 + \alpha}\}$ (see Section~\ref{explain-rounding-name}), and this is the reason we use the term ``rounding'' to describe our corrective procedure to get $w$ to satisfy \eqref{ineqgradcondforprog-early} and the term ``rounding condition'' to refer to \eqref{ineqgradcondforprog-early}. 

We develop two versions of rounding: a parallel method and a  sequential one that has an improved dependence on $p$. Each version is based on the principles that (1) one can increase those entries of $w$ at which the rounding condition  \eqref{ineqgradcondforprog-early} does not hold \emph{while decreasing} the objective value, and (2)  the vector $w$ obtained after this update  is closer to satisfying \eqref{ineqgradcondforprog-early}. 

We believe that such a principle of identifying a technical condition needed for fast convergence and the accompanying rounding procedures could be useful in other optimization problems. Additionally, we develop Algorithm~\ref{alg_oneloop}, which, by varying the step sizes in the update rule, maintains \eqref{ineqgradcondforprog-early}  as invariant, thereby eliminating the need for a separate rounding and progress steps.

\subsection{Applications and Related Work}\label{sec:appsRelatedWork} We elaborate here on the applications of Lewis weights we briefly alluded to in Section~\ref{sec:intro}.  While for many applications  (such as  pre-processing in optimization  \cite{CohenPeng2015}) approximate weights suffice, solving regularized $D$-optimal and computing $\tilde{O}(n)$ self-concordant barriers to high precision do use high precision Lewis weights. 

\paragraph{Pre-processing in optimization.} Lewis weights are used as scores to sample rows of an input tall data matrix so the $\ell_p$ norms of the product of the matrix with vectors are preserved. They have been used in row sampling algorithms for data pre-processing \cite{drineas2006sampling, drineas2012fast, li2013iterative, DBLP:conf/innovations/CohenLMMPS15}, for computing dimension-free strong coresets for $k$-median and subspace approximation \cite{sohler2018strong}, and for fast tensor factorization in the streaming model \cite{chhaya2020streaming}.  Lewis weights are also used for  $\ell_1$  regression, a popular model in machine learning used to capture robustness to outliers, in:  \cite{durfee2018ell_1} for stochastic gradient descent pre-conditioning, \cite{pmlr-v119-li20o} for quantile regression, and \cite{braverman2020near} to provide algorithms for linear algebraic problems in the sliding window model.

\paragraph{John ellipsoid and D-optimal design.} 
As noted in Remark~\ref{remarkJohnEll}, a fast algorithm for  Lewis weights could yield faster algorithms for computing John ellipsoid, a problem with a long history of work \cite{khachiyan1996rounding, sun2004computation, kumar2005minimum, damla2008linear, CohenCousinsLeeYang2019, zhao2020analysis}. It is known \cite{Todd2016} that the John ellipsoid problem is dual to the (relaxed) D-optimal experiment design problem \cite{10.5555/1137748}. D-optimal design seeks to select a set of linear experiments with the largest confidence ellipsoid for its least-square estimator \cite{pmlr-v70-allen-zhu17e, pmlr-v99-madan19a, singh2020approximation}.

Our problem \eqref{def_objective} is equivalent to  $\frac{p}{p-2}$-regularized D-optimal design, which can be interpreted as enforcing a polynomial experiment cost: viewing $w_i$ as the fraction of  resources allocated to experiment $i$, each  $w_i$ is penalized by $w_i^{\frac{p}{p-2}}$. This regularization also appears in fair packing and fair covering problems  \cite{marasevic2016fast, DBLP:journals/siamjo/DiakonikolasFO20} from operations research.

\paragraph{Self-concordance.} 
Self-concordant barriers are fundamental in convex optimization \cite{nesterov1994interior}, combinatorial optimization \cite{LeeSidford2014}, sampling \cite{10.1145/1536414.1536491, 10.1145/3357713.3384272},
and online learning \cite{abernethy2008competing}. Although there are (nearly) optimal self-concordant barriers for any convex set \cite{nesterov1994interior, bubeck2015entropic, lee2018universal}, computing them involves sampling from log-concave distributions, itself an expensive process with a \polyeps runtime. \cite{LeeSidford2014} shows how to construct nearly optimal barriers for polytopes using Lewis weights. Unfortunately, doing so still requires polynomial-many steps to compute these weights; \cite{LeeSidford2014} bypass this issue by showing it suffices to work with Lewis weights for $p \approx 1$. In this paper, we  show how to compute Lewis weights by computing leverage scores of polylogarithmic-many matrices. This gives the first nearly optimal self-concordant barrier for polytopes that can be evaluated to high accuracy with  depth polylogarithmic in the dimension.

\begin{theorem}[Applying Theorem~\ref{thm_main_thm} to {\cite[Section 5]{DBLP:journals/corr/abs-1910-08033}}]\label{selfconcordancethm}
Given a non-empty polytope $P=\{x\in\mathbb{R}^{n}~|~\mA x>b\}$ for 
full rank $\mA\in\mathbb{R}^{m\times n}$, there is a $O(n\log^{5}m)$-self
concordant barrier $\psi$ for $P$ such that for any $\epsilon>0$
and $x\in P$, in $O(mn^{\omega-1}\log^{3}m\log(m/\epsilon))$-work
and $O(\log^{3}m\log(m/\epsilon))$-depth, we can compute $g\in\mathbb{R}^{n}$
and $\mathbf{H}\in\mathbb{R}^{n\times n}$ with $\|g-\nabla\psi(x)\|_{\nabla^{2}\psi(x)^{-1}}\leq\epsilon$
and $\nabla^{2}\psi(x)\preceq\mathbf{H}\preceq O(\log m)\nabla^{2}\psi(x)$.
With an additional $O(m^{\omega+o(1)})$ work, $\mathbf{H}\in\mathbb{R}^{n\times n}$ with $(1-\epsilon)\nabla^{2}\psi(x)\preceq\mathbf{H}\preceq O(1+\epsilon)\nabla^{2}\psi(x)$ can be computed as well. 
\end{theorem} 

%% file: our_analysis.tex
\subsection{Notation and Preliminaries}\label{sec:notation}
We use $\mA$ to denote our full-rank $m \times n$ ($m \geq n$) real-valued input matrix and $\lw\in \R^m$ to denote the vector of Lewis weights of $\mA$, as defined in \eqref{eq_def_lewis_weights} and \eqref{def_objective}. All matrices appear in boldface uppercase and vectors in lowercase. For any vector (say, $\sigma$), we use its uppercase boldfaced form ($\mathbf{\Sigma}$) to denote the diagonal matrix $\mathbf{\Sigma}_{ii} = \sigma_i$. For a matrix $\mathbf{M}$, the matrix $\mathbf{M}^{(2)}$ is the Schur product (entry-wise product) of $\mathbf{M}$ with itself. For matrices $\mathbf{A}$ and $\mathbf{B}$, we use $\mathbf{A}\succeq \mathbf{B}$  to mean $\mathbf{A} - \mathbf{B}$ is positive-semidefinite. For vectors $a$ and $b$, the inequality $a\leq b$ applies entry-wise.  We use $e_i$ to denote the $i$'th standard basis vector. We define $[n] \defeq \{1, 2, \dotsc, n\}$. As in \eqref{def_objective}, since we defined $\alpha \defeq \frac{2}{p - 2}$, the ranges of $p \in (2, 4)$ and $p \geq 4$ translate to $\alpha > 1$ and $\alpha \in (0, 1]$, respectively. From hereon, we work with $\alpha$. We also define $\alphamax = \max\{1, \alpha\}$. For a matrix $\mA\in \R^{m\times n}$ and  $w\in \R^m_{> 0}$,  we define the projection matrix $\mP(w) \defeq \mW^{1/2} \mA (\mA^\top \mW \mA)^{-1} \mA^{\top} \mW^{1/2} \in \R^{m \times m}$. The quantity $\mP(w)_{ii}$ is precisely the leverage score of the $i$'th row of $\mW^{1/2} \mA$: \[\sigma_i(w) \defeq w_i \cdot a_i^\top (\mA^\top  \mW \mA)^{-1} a_i. \numberthis\label{def-levscoreW}\]

\begin{fact}[\cite{LeeSidford2014}]\label{fact_projmatrices} For all $w\in \R_{> 0}^m$ we have that $0\leq \sigma_i(w)\leq 1$ for all $i\in [m]$, $\sum_{i\in [m]} \sigma_i (w) \leq n$, and $\mathbf{0}\preceq \mP(w)^{(2)} \preceq \mathbf{\Sigma}(w)$. 
\end{fact} 
\section{Our Algorithm} We present  Algorithm~\ref{alg_ourAlg} to compute an $\teps$-additive solution to \eqref{def_objective}. We first provide the following definitions that we frequently refer to in our algorithm and analysis.  Given $\alpha>0$ and $w\in \R^m_{> 0}$,  the $i$'th coordinate of $\rho(w) \in \R^m$ is \[ \rho_i(w) \defeq \frac{\sigma_i(w)}{w_i^{1+\alpha}}. \numberthis\label{def-rho-notation}\]  Based on this quantity, we define the following procedure, derived from approximating a quasi-Newton update on the objective $\fobj$ from \eqref{def_objective}: \[
\left[
\progcode(w, \coordset, \eta)
\right]_i
\defeq w_i\left[1  + \eta_i \cdot  \frac{\rho_i(w)-1}{\rho_i(w) + 1}\right] \text{ for all } i \in \coordset \subseteq \{1, 2, \dotsc, m\}.
\numberthis\label{eq-wratiostep}
\]  Using these definitions, we can now describe our algorithm. Depending on whether the following condition (``rounding condition'') holds, we run either $\progcode({}\cdot{})$  or  $\roundcode({}\cdot{})$ in each iteration:
\[  \rho_{\max}(w) \defeq  \max_{i \in [m]} \rho_i(w) \leq 1 + \alpha. \numberthis\label{eq_grad_cond_for_prog}\] 
Specifically, if \eqref{eq_grad_cond_for_prog} is \emph{not} satisfied,  we run $\roundcode({}\cdot{})$, which returns a vector that \emph{does} satisfy it without increasing the objective value. We design two versions of $\roundcode({}\cdot{})$, one parallel  (Algorithm~\ref{alg_roundparallel}) and one sequential (Algorithm~\ref{alg_roundsequential}), with the sequential algorithm having an improved dependence on $\alpha$, to update the coordinates violating 
 \eqref{eq_grad_cond_for_prog}. We apply one extra step of rounding to the vector returned after $\totaliters$ iterations of Algorithm~\ref{alg_ourAlg} and transform it appropriately to obtain our final output. In the following lemma (proved in Section~\ref{sec:opti-to-lw}), we justify that this  output is indeed the solution to \eqref{def-lewiswts-obj-approxLW}. 
\begin{restatable}[Lewis Weights from Optimization Solution]{lemmma}{LEMroundingToLW}\label{lem-fn-to-LW}
Let $w\in \R_{>0}^m$ be a vector at which the objective \eqref{def_objective} is $\teps$-suboptimal in the additive sense for $\teps = \tepsval$, i.e., $\fobj(\wopt)\leq \fobj(w) \leq \fobj(\wopt) + \teps$. Further assume that $w$ satisfies the rounding condition: $\rho_{\max}(w) \leq 1+ \alpha$. Then, the vector $\hw$ defined as $\hw_i = (a_i^\top (\mA^\top \mW \mA)^{-1} a_i)^{1 / \alpha}$ satisfies $\hw_i \approx_{\eps} \lw_i$ for all $i\in [m]$, thus achieving the goal spelt out in \eqref{def-lewiswts-obj-approxLW}. 
\end{restatable}
\begin{algorithm}[!ht]
\DontPrintSemicolon
\caption{\textbf{Lewis Weight Computation Meta-Algorithm}}\label{alg_ourAlg}
\KwIn{Matrix $\mA \in \reals^{m \times n}$, parameter $p > 2$, accuracy $\eps$}
\KwOut{Vector $\hw\in \reals^m_{>0}$ that satisfies \eqref{def-lewiswts-obj-approxLW}}

For all $i \in [m]$, initialize $w_i^{(0)} = \frac{n}{m}$.\; 

Set $\alpha = \frac{2}{p-2}$, $\alphamax = \max(\alpha, 1)$, $\teps = \tepsval$, and $\totaliters = \Ord(\max(\alpha^{-1}, \alpha)\log (m / \teps))$.\;  

\For{$k = 1, 2, 3, \dotsc, \totaliters$}{

    $\widetilde{w}^{(k)} \leftarrow \roundcode(w^{(k - 1)}, \mA, \alpha)$ \Comment{ Invoke Algorithm~\ref{alg_roundparallel} (parallel) or \ref{alg_roundsequential}} (sequential)
    \BlankLine

    $w^{(k)} \leftarrow \progcode(\widetilde{w}^{(k)}, [m], \tfrac{1}{3\alphamax}\1)$
    \Comment{See \eqref{eq-wratiostep} and Lemma \ref{lem_fast_steps_count}} \;
}	

Set $\wround \leftarrow \roundcode(w^{(\totaliters)}, \mA, \alpha)$ \label{line:meta-lastround} \;

Return $\hw \in \R^m_{>0}$, where $\hw_i = (a_i^\top (\mA^\top \mwround \mA)^{-1} a_i)^{1/\alpha}$. \Comment{See Section~\ref{sec:opti-to-lw}} \label{line:finalLW} \; 

\end{algorithm}	 

\begin{algorithm}[!ht]
\DontPrintSemicolon
\caption{\roundparallelcode($w$, $\mA$, $\alpha$) }\label{alg_roundparallel}
\KwIn{Vector $w\in \reals^m_{>0}$, matrix $\mA\in \reals^{m \times n}$, parameter $\alpha > 0$}
\KwOut{Vector $w\in \reals^m_{>0}$ satisfying \eqref{eq_grad_cond_for_prog}}

Define $\rho(w)$ as in \eqref{def-rho-notation}\; 

\While{$\coordset = \{i ~|~ \rho_i(w) > 1+\alpha \} \neq \emptyset$  }{\label{line:parallel_while_start} 

    $w \leftarrow \progcode(w, \coordset, \tfrac{1}{3\alphamax}\1)$
  \Comment{See Section~\ref{sec-slowiters-ana}} \; 
}

Return $w$ \; 

\end{algorithm}

\begin{algorithm}[!ht]
\DontPrintSemicolon
\caption{\roundsequentialcode($w$, $\mA$,$\alpha$)}\label{alg_roundsequential}
\KwIn{Vector $w\in \reals^m_{>0}$, matrix $\mA\in \reals^{m \times n}$, parameter $\alpha > 0$}
\KwOut{Vector $w\in \reals^m_{>0}$ satisfying \eqref{eq_grad_cond_for_prog}}
Define $\rho(w)$ as in \eqref{def-rho-notation} and $\sigma(w)$ as in \eqref{def-levscoreW}\; 

Define $\coordset = \{i ~|~ \rho_i(w) \geq 1 \}$ \; 

\For{ $i \in \coordset$}{
 $w_i\leftarrow w_i(1+\delta_i)$, where $\delta_i$ solves $\rho_{i}(w)=(1+ \delta_i \sigma_i(w)) (1+\delta_i)^{\alpha}$  \Comment{see Section~\ref{sec:anaseqalg}} \; 
}

Return $w$ \; 
\end{algorithm}

\subsection{Analysis of $\progcode({}\cdot{})$}\label{sec:descentanalysis}
We first analyze $\progcode({}\cdot{})$ since it is common to both the parallel and sequential algorithms.

\begin{lemma}[Iteration Complexity of $\progcode({}\cdot{})$] \label{lem_fast_steps_count} Each iteration of  $\progcode({}\cdot{})$  (described in \eqref{eq-wratiostep}) decreases the value of $\fobj$. Assuming that  $\roundcode({}\cdot{})$ does not increase the value of the objective in \eqref{def_objective}, for any given accuracy parameter $0 < \teps <1$, the number of $\progcode({}\cdot{})$ steps that Algorithm~\ref{alg_ourAlg} performs before achieving $\fobj(w)\leq \fobj(\wopt) +\teps$ is given by the following bound: \[\totaliters = \Ord(\max( \alpha^{-1}, \alpha) \log (m/  \teps)).\]  
\end{lemma}

As is often the case to obtain such an iteration complexity, we prove Lemma~\ref{lem_fast_steps_count} by incorporating the maximum sub-optimality in function value (Lemma~\ref{lem_subopt}) and the initial error bound (Lemma~\ref{lem_initialError}) into the inequality describing minimum function progress (Lemma~\ref{lem:parallel_progress}). The assumption that $\roundcode({}\cdot{})$ does not increase the value of the objective is justified in Lemma~\ref{lem:parallel_iteration}. 

Since our algorithm leverages quasi-Newton steps, we begin our analysis by stating the gradient and Hessian of the objective in \eqref{def_objective} as well as the error at the initial vector $w^{(0)}$, as measured against the optimal function value. The Hessian below is positive semidefinite when $\alpha \geq 0$ (equivalently, when $p\geq 2$) and not necessarily so otherwise. Consequently, the objective is convex for $\alpha \geq 0$,  and we therefore consider only this setting throughout.

\begin{restatable}[Gradient and Hessian]{lemmma}{LEMgradAndHess}\label{lem_gradAndHess}
For any $w \in \R^m_{>0}$,  the objective in  \eqref{def_objective}, $\fobj(w) = -\log \det (\mA^\top \mW \mA) + \frac{1}{1+\alpha}\ones^\top w^{1+\alpha}$, has gradient and Hessian given by the following expressions.  
$$
\left[ \nabla \fobj(w)\right]_i = w_i^{-1} \cdot (w_i^{1+\alpha} - \sigma_i(w)) \text{ and } \nabla^2 \fobj(w) = \mW^{-1} \mP(w)^{(2)} \mW^{-1} + \alpha \mW^{\alpha-1}.
$$ 
\end{restatable}

\begin{restatable}[Initial Sub-Optimality]{lemmma}{LEMinitialError}\label{lem_initialError}
At the start of Algorithm~\ref{alg_ourAlg}, the value of the  objective of  \eqref{def_objective} differs from the optimum objective value as $\fobj(w^{(0)}) \leq \fobj(\wopt) + n \log (m/n)$. 
\end{restatable} 
\subsubsection{Minimum Progress and Maximum Sub-optimality in  $\progcode({}\cdot{})$}
We first prove an upper bound on objective sub-optimality, necessary to obtain a runtime polylogarithmic in $1/\eps$. Often, to obtain such a rate, the bound involving objective sub-optimality has a quadratic term derived from the Hessian; our lemma is somewhat non-standard in that it uses only the convexity of $\fobj$. Note that this lemma crucially uses  \eqref{eq_grad_cond_for_prog}.  

\begin{restatable}[Objective Sub-optimality]{lemmma}{LEMsuboptkeylemma}\label{lem_subopt}
Suppose $w \in \reals^m_{> 0}$ and $\rho_{\max}(w) \leq 1 + \alpha$. Then the  value of the objective of \eqref{def_objective} at $w$ differs from the optimum objective value as follows. 
\begin{align*}
\fobj(w) - \fobj(\wopt)
&\leq  \sum_{i \in [m]}
\frac{w_i^{1+ \alpha}}{1 + \alpha} \left(1 + \frac{\rho_i(w)}{\alpha}\right)
 \left( \rho_i(w) - 1 \right)^2  \leq  5 \max\{1, \alpha^{-1}\}
\sum_{i \in [m]}
w_i^{1+ \alpha} \frac{(\rho_i(w) - 1)^2}{\rho_i(w) + 1}. 
\end{align*}
\end{restatable} 

\begin{proof}
Since $g(w) \defeq - \log \det(\mA^\top \mW \mA)$ is convex  and $[\nabla g(w)]_i = - w_i^{-1} \sigma_i(w)$, we have 
\[
g(\wopt) \geq g(w) + \nabla g(w)^\top (\wopt - w)
= g(w) + \sum_{i \in [m]} \left(- \frac{\sigma_i(w) \wopt_i}{w_i}  + \sigma_i(w))\right),
\]
and therefore, 
\begin{align*}
\fobj(\wopt) - \fobj(w) 
&= g(\wopt) - g(w) + \frac{1}{1+ \alpha} \sum_{i \in [m]} \left([\wopt]_i^{1+ \alpha} - w_i^{1 + \alpha}\right) \\
&\geq \sum_{i \in [m]} c_i
\text{ where } c_i \defeq - \frac{\sigma_i(w) \wopt_i}{w_i} + \sigma_i(w) + \frac{1}{1+ \alpha}  \left([\wopt]_i^{1+ \alpha} - w_i^{1 + \alpha}\right)\,.
\end{align*} To prove the claim, it suffices to bound each $c_i$ from below. First, note that
\begin{align}
c_i 
&\geq \min_{v \geq 0}  
-\frac{v \cdot \sigma_i(w)}{w_i} + \sigma_i(w) + \frac{1}{1+ \alpha}  \left(v^{1+ \alpha} - w_i^{1 + \alpha}\right) = - \frac{\alpha}{1 + \alpha} \left(\frac{\sigma_i(w)}{w_i}\right)^{1 + \frac{1}{\alpha}}
+ \sigma_i(w) - \frac{w_i^{1 + \alpha}}{1 + \alpha}
\nonumber
\\
&= \frac{w_i^{1+ \alpha}}{1 + \alpha} \left[ - \alpha \rho_i(w)^{1 + \frac{1}{\alpha}} + (1 + \alpha) \rho_i(w) - 1 \right] 
\label{eq:cbound_1}
\end{align}
where  the first equality  used  that the minimization problem is convex and the solutions to $-\sigma_i(w) w_i^{-1} + v^\alpha = 0$ (i.e.\ where the gradient is 0) is a minimizer, and the second equality  used  $\rho_i(w) = \sigma_i(w) / w_i^{1 + \alpha}$. Applying  $\rho_i(w) \leq 1+ \alpha$, $1 + x \leq \exp x$, and $\exp x \leq 1 + x + x^2$ for $x \leq 1$ yields  
\begin{align*}
\rho_i(w)^{\frac{1}{\alpha}} &= 
(1 - (1 - \rho_i(w)))^{\frac{1}{\alpha}} \leq  \exp(\tfrac{1}{\alpha}(\rho_i(w) - 1)) \leq 1 + \frac{1}{\alpha}(\rho_i(w) - 1) + \frac{1}{\alpha^2}(\rho_i(w) - 1)^2.
\numberthis\label{eq:cbound_2}
\end{align*}
Combining \eqref{eq:cbound_2} with \eqref{eq:cbound_1} yields that 
\begin{align*}
c_i &\geq 
\frac{w_i^{1+ \alpha}}{1 + \alpha} \left[ - \alpha \rho_i(w) \left[
1 + \left(\frac{\rho_i(w) - 1}{\alpha}\right) + \left(\frac{\rho_i(w) - 1}{\alpha}\right)^2
\right]
 + (1 + \alpha) \rho_i(w) - 1 \right] \\
&= 
\frac{w_i^{1+ \alpha}}{1 + \alpha} \left[-1 + 2 \rho_i(w) - \rho_i(w)^2 - \frac{\rho_i(w)}{\alpha} \cdot (\rho_i(w) - 1)^2 \right] 
= 
-
\frac{w_i^{1+ \alpha}}{1 + \alpha} \left(1 + \frac{\rho_i(w)}{\alpha}\right)
\cdot \left( \rho_i(w) - 1 \right)^2
\end{align*}
The claim then follows from the  fact that for $\rho_i(w) \leq 1 + \alpha$, we have $\frac{(1 + \rho_i(w) \alpha^{-1}) (1 + \rho_i(w))}{1 + \alpha}
\leq \frac{1}{1 + \alpha}+ \frac{1}{\alpha} + 1 + 1 + \frac{1}{\alpha}
\leq 5 \max\{1, \alpha^{-1}\}$. 
\end{proof}

\begin{restatable}[Function Decrease in $\progcode({}\cdot{})$]{lemmma}{LEMminfndecrease}\label{lem:parallel_progress}
Let $w, \eta \in \R^m_{> 0}$ with $\eta_i \in [0, \tfrac{1}{3\alphamax}]$  for all $i \in [m]$. Further, let $w^+ = \progcode(w, [m], \eta)$, where $\progcode$ is defined in \eqref{eq-wratiostep}. 
Then, $w^+ \in \R^m_{>0}$ with the following decrease in function objective. \[\fobj(w^+) \leq  \fobj(w)  - \sum_{i \in [m]} \frac{\eta_i}{2} \cdot w_i^{1 + \alpha} \cdot \frac{(\rho_i(w) - 1)^2}{\rho_i(w) + 1}.\] 
\end{restatable}
The proof of this lemma resembles that of quasi-Newton method: first, we write a second-order Taylor approximation of $\fobj(w^+)$ around $w$ and apply  Fact~\ref{fact_projmatrices} to Lemma~\ref{lem_gradAndHess} to obtain the upper bound on Hessian: $
\nabla^2 \fobj(\wmid) = \mwmid^{-1} \mP(\wmid)^{(2)} \mwmid^{-1}  +\alpha \mwmid^{\alpha-1} 
\preceq \mwmid^{-1} \Sigma(\wmid) \mwmid^{-1} + \alpha \mwmid^{\alpha-1}.$ We further use the expression for $\nabla \fobj(w)$ in this second-order approximation and simplify to obtain the claim, as detailed in Section~\ref{sec:technicalproofsapp}.

\subsubsection{Iteration Complexity of $\progcode({}\cdot{})$}\label{sec:roundparallelanalysis}
\begin{proof}[Proof of Lemma~\ref{lem_fast_steps_count}] Since Algorithm~\ref{alg_ourAlg} calls $\progcode({}\cdot{})$ after running $\roundcode({}\cdot{})$, the requirement $\rho_{\max}(w)\leq 1+\alpha$ in Lemma~\ref{lem_subopt} is met. Therefore, we may combine Lemma~\ref{lem_subopt} alongwith Lemma~\ref{lem:parallel_progress} and our choice of $\eta_i = \frac{1}{3\alphamax}$ in  Algorithm~\ref{alg_ourAlg} to get a geometric decrease in function error as follows.   
\begin{align*}
\fobj(w^{+})-\fobj(\wopt) & \leq \fobj(w)-\fobj(\wopt)-\frac{1}{6 \max(\alpha, 1)}\sum_{i=1}^{m}w_i^{1+\alpha} \frac{(\rho_{i}(w)-1 )^{2}}{\rho_{i}(w)+1 }\\
 & \leq\left( 1-\frac{1}{30 \max(1,\alpha) \cdot \max(1,\alpha^{-1})} \right)(\fobj(w)-\fobj(\wopt)). \numberthis\label{fastitercount-1}
\end{align*} 
We apply this inequality recursively over all iterations of Algorithm~\ref{alg_ourAlg}, while also using  the assumption that $\roundcode({}\cdot{})$ does not increase the objective value. Setting the final value of \eqref{fastitercount-1} to $\teps$,  bounding the initial error as $\fobj(w)-\fobj(\wopt)\leq n \log (m/n)\leq m^2$ by Lemma~\ref{lem_initialError},  observing  $\max(1,\alpha) \cdot \max(1,\alpha^{-1}) = \max(\alpha, \alpha^{-1})$, and taking logarithms on both sides of the inequality  gives the claimed iteration complexity of $\progcode({}\cdot{})$. 
\end{proof}
\section{Analysis of $\roundcode({}\cdot{})$: The Parallel Algorithm}\label{sec-slowiters-ana}
The main export of  this section is the  proof of our main theorem about the parallel algorithm (Theorem~\ref{thm_main_thm}). This proof combines the iteration count of $\progcode({}\cdot{})$ from the preceding section with the analysis of Algorithm~\ref{alg_roundparallel} (invoked by $\roundcode({}\cdot{})$ in the parallel setting), shown next. In Lemma~\ref{lem:parallel_iteration}, we show that $\roundparallelcode({}\cdot{})$ decreases the function objective, thereby justifying the key assumption in Lemma~\ref{lem_fast_steps_count}. Lemma~\ref{lem:parallel_iteration} also shows an upper bound on the new value of $\rho$ after one \texttt{while} loop of $\roundparallelcode({}\cdot{})$, and by combining this with the maximum value of $\rho$ for termination in Algorithm~\ref{alg_roundparallel}, we get the iteration complexity of $\roundparallelcode({}\cdot{})$ in Corollary~\ref{cor:roundparallelcodeiters}. 
\begin{lemma}[Outcome of $\roundparallelcode({}\cdot{})$]\label{lem:parallel_iteration}Let  $w^+\in \reals^m_{>0}$ be the state of $w\in \reals^m_{>0}$ at the end of one $\texttt{while}$ loop of $\roundparallelcode({}\cdot{})$ (Algorithm~\ref{alg_roundparallel}). Then,  $\fobj(w^+) \leq \fobj(w)$ 
and $\rho_{\max}(w^+) \leq  (1+\frac{\alpha}{3\alphamax(2+\alpha)})^{-\alpha} \rho_{\max}(w)$.
\end{lemma}
\begin{proof} Each iteration of the \texttt{while} loop in $\roundparallelcode({}\cdot{})$  performs $\progcode(w, \coordset, \frac{1}{3\alphamax}\1)$ over the set of coordinates $\coordset = \{i: \rho_i(w) > 1+ \alpha\}$. Lemma~\ref{lem:parallel_progress} then immediately proves $\fobj(w^+)\leq \fobj(w)$, which is our first claim. 

To prove the second claim, note that in Algorithm~\ref{alg_roundparallel}, for every $i\in \coordset$ 
\begin{align*}
w^+_{i} &= w_i + \frac{w_{i}}{3\alphamax}\cdot\left[\frac{\rho_{i}(w)-1}{\rho_{i}(w)+1}\right] 
\geq w_i +\frac{w_{i}}{3\alphamax}\cdot\left[\frac{\alpha}{1+1+\alpha}\right]= w_{i}\cdot\left(1+\frac{\alpha}{3\alphamax(2+\alpha)}\right),
\end{align*} where the second step is by the monotonicity of $x\rightarrow \frac{x-1}{x+1}$ for $x\geq 1$. Combining it with $w_i^+ = w_i$ for all $i \notin \coordset$ implies that $w^+ \geq w$. Therefore, for all $i \in \coordset$, we have  

\begin{align*}
\rho(w^+)_i 
&= [w^+_i]^{-\alpha} [\mA (\mA^\top \mW^+ \mA)^{-1} \mA^\top]_{ii} 
\leq \left[ 1+\frac{\alpha}{3\alphamax(2+\alpha)} \right]^{-\alpha} \cdot w_i^{-\alpha} [\mA (\mA^\top \mW \mA)^{-1} \mA^\top]_{ii}. \numberthis\label{roundpar-rho-ub}
\end{align*}
\end{proof} 
\begin{corollary}\label{cor:roundparallelcodeiters}
Let $w$ be the input to $\roundparallelcode({}\cdot{})$. Then, the number of iterations of the \texttt{while} loop of $\roundparallelcode({}\cdot{})$  is at most $O\left((1+\alpha^{-2}) \log(\frac{\rho_{\max}(w) }{ 1 + \alpha}) \right)$.
\end{corollary}
\begin{proof}
Let $w^{(i)}$ be the value of $w$ at the start of the $i$'th execution of the \texttt{while} loop of $\roundparallelcode({}\cdot{})$. Repeated application of Lemma~\ref{lem:parallel_iteration} over $k$ executions of the \texttt{while} loop gives $
\rho_{\max}(w^{(k)}) \leq \rho_{\max}(w) \left(1+\frac{\alpha}{3\alphamax(2+\alpha)}\right)^{-\alpha k} $. We set $\rho_{\max}(w) \left(1+\frac{\alpha}{3\alphamax(2+\alpha)}\right)^{-\alpha k}  \leq  1 + \alpha$ in accordance with the termination condition of $\roundparallelcode({}\cdot{})$. Next, applying $1+x\leq \exp(x)$, and taking logarithms on both sides yields the claimed limit on the number of iterations, $k$. 
\end{proof}
\begin{lemma}\label{lem-total-slowiters} 
 Over the entire run of Algorithm~\ref{alg_ourAlg}, the  \texttt{while} loop of $\roundparallelcode({}\cdot{})$ runs for at most  $O\left(\totaliters \cdot\alpha^{-2} \cdot \log\left(\tfrac{m}{n(1+\alpha)}\right) \right)$ iterations if $\alpha\in (0, 1]$ and $\Ord\left(\totaliters \cdot \alpha  \cdot\log\left( \tfrac{m}{n(1+\alpha)}\right) \right)$ iterations if $\alpha\geq 1$. 
\end{lemma}

\begin{proof}
Note that $\rho_{\max}(\frac{n}{m}) \leq (\frac{m}{n})^{1 + \alpha}$;  consequently, in the first iteration of Algorithm~\ref{alg_ourAlg}, there are at most $O((\alpha+\alpha^{-2}) \log(m /(n(1+\alpha))))$ iterations of the \texttt{while} loop of $\roundparallelcode({}\cdot{})$ by Corollary~\ref{cor:roundparallelcodeiters}. Note that between each call to $\roundparallelcode({}\cdot{})$, for all $i\in [m]$, 
\begin{align*}
w^+_{i} &= w_i + \frac{w_{i}}{3\alphamax}\cdot\left[\frac{\rho_{i}(w)-1}{\rho_{i}(w)+1}\right] 
\geq w_i +\frac{w_{i}}{3\alphamax}\cdot\left[\frac{-1}{1+1+\alpha}\right]= w_{i}\cdot\left(1-\frac{1}{(3\alphamax)(2+\alpha)}\right),
\end{align*}
where the first inequality is by using the fact that the output $w$ of $\roundparallelcode({}\cdot{})$ satisfies $\rho_{\max}(w)\leq 1+ \alpha$. Therefore, applying the same logic as in \eqref{roundpar-rho-ub}, we get that between two calls to $\roundparallelcode({}\cdot{})$, the value of 
$\rho_i(w)$ increases by at most $\left(1-\frac{1}{(3\alphamax)(2+\alpha)}\right)^{-(1+\alpha)} = O(1)$ for all $i\in [m]$. Combining this with Corollary~\ref{cor:roundparallelcodeiters} and  the total initial iteration count and observing that  $\totaliters$ is the total number of calls to $\roundparallelcode({}\cdot{})$   finishes the proof. 
\end{proof}

\subsection{Proof of Main Theorem (Parallel)}\label{sec:pf-main-thm-par}

\begin{proof}(Proof of Theorem~\ref{thm_main_thm})  First, we show correctness. Note that, as a corollary of Lemma \ref{lem_fast_steps_count},   $\fobj(w^{(\totaliters)}) \leq \fobj(\wopt) + \teps$. By the properties of $\roundcode$ as shown in Lemma~\ref{lem:parallel_iteration}, we also have that $\fobj(\wround) \leq \fobj(\wopt) + \teps$ and $\rho_{\max}(\wround)\leq 1+\alpha$ for $\wround = \roundcode(w^{(\totaliters)}, \mA, \alpha)$.  Therefore, Lemma~\ref{lem-fn-to-LW} is applicable, and by the choice of $\teps = \tfrac{\alpha^4 \eps^4}{(2 m (\sqrt{n} + \alpha) (\alpha + \alpha^{-1}))^4}$, we conclude that $\hw\in \R^m$ defined as $\hw_i = (a_i^\top (\mA^\top \mwround \mA)^{-1} a_i)^{1/\alpha}$  satisfies $\hw_i \approx_{\eps} \lw_i$ for all $i \in [m]$. Combining the iteration counts of $\progcode({}\cdot{})$ from Lemma~\ref{lem_fast_steps_count} and of $\roundparallelcode({}\cdot{})$ from Lemma~\ref{lem-total-slowiters} yields the total iteration count as $O(\alpha^{-3} \log (m/(\eps\alpha)))$ if $\alpha \leq 1$ and $O(\alpha^2 \log (m/\eps))$ if $\alpha > 1$. As stated in Section~\ref{sec:notation},  $\alpha = \frac{2}{p - 2}$, and so translating these rates in terms of $p$ gives $O(p^3 \log (mp/\eps))$ for $p \geq 4$ and $O(p^{-2} \log (mp/\eps))$ for $p \in (2, 4)$, thereby proving the stated claim. The cost per iteration is  $O(mn^2)$\footnote{This can be improved to $O(mn^{\omega-1})$ using fast matrix multiplication.}  for multiplying two $m\times n$ matrices.  
\end{proof}

\section{Analysis of $\roundcode({}\cdot{})$: Sequential Algorithm} \label{sec:anaseqalg}
We now analyze Algorithm~\ref{alg_roundsequential}. Note that these proofs work for all $\alpha > 0$. 

\begin{restatable}[Coordinate Step Progress]{lemmma}{LEMcoordinateprogress}\label{lem_coordinate_progress_1}Given $w\in\R_{>0}^{m}$, a coordinate $i\in [m]$, and $\delta_i\in \R$,
we have
 
\[
\fobj(w+\delta_i w_{i}e_{i})=\fobj(w)-\log(1+\delta_i\sigma_{i}(w))+\frac{w_{i}^{1+\alpha}}{1+\alpha}((1+\delta_i)^{1+\alpha}-1).
\]

\end{restatable}
\begin{proof}
By definition of $\fobj$, we have 
\begin{align*}
\fobj(w+\delta_i w_{i}e_{i})= & -\log\det(\mA^{\top}\mW\mA+\delta_i w_{i}a_{i}a_{i}^{\top})+\frac{1}{1+\alpha}\sum_{j\neq i}w_{j}^{1+\alpha}+\frac{w_{i}^{1+\alpha}}{1+\alpha}(1+\delta_i)^{1+\alpha}.
\end{align*} Recall the matrix determinant lemma: $\det(\mA + uv^\top) = (1+v^\top \mA^{-1} u) \det(\mA)$. Applying it to $\det(\mA^\top \diag (w + \delta_i w_i e_i) \mA)$ in the preceding expression for $\fobj(w+\delta_i w_i e_i)$ proves the lemma. 

\end{proof}

\begin{restatable}[Coordinate Step Outcome]{lemmma}{LEMcoordinateprogressmore}\label{lem_coordinate_progress_2}Given $w\in\R_{>0}^{m}$
and  $\coordset = \{i: \rho_i(w) \geq 1 \}$, let $w^+ = w+\delta_i w_i e_i$  for any  $i\in \coordset$, where $\delta_i =\arg\min_{\delta} \left[- \log (1+\delta \sigma_i(w)) + \frac{1}{1+\alpha} w_i^{1+\alpha} ((1+\delta)^{1+\alpha}-1)\right]$.  Then, we have  
 $\fobj(w^+) \leq \fobj(w)$ and  $\rho_i(w^+) \leq 1$.

\end{restatable} 
\begin{proof}
We note that $\min_{\delta} \left[ - \log (1+\delta \sigma_i(w)) + \frac{1}{1+\alpha} w_i^{1+\alpha} ((1+\delta)^{1+\alpha}-1)\right] \leq 0$. Then, Lemma~\ref{lem_coordinate_progress_1} implies the first claim. Since the update rule optimizes over $\fobj$ coordinate-wise, at each step the optimality condition given  by $\rho_i(w^+) =1$ is met for each $i\in \coordset$. The second claim is then proved by noting that for $j\neq i$, $w_j^+ = w_j$ and by the Sherman-Morrison-Woodbury identity, $\rho_j(w^+) \leq \rho_j(w)$:  \[a_j^\top (\mA^\top \mW^{+} \mA)^{-1} a_j 
= a_j^\top (\mA^\top \mW \mA)^{-1} a_j - \delta_i w_i \frac{(a_j^\top (\mA^\top \mW \mA)^{-1} a_j)^2}{1 + \delta_i w_i a_i^\top (\mA^\top \mW \mA)^{-1} a_i } \leq a_j^\top (\mA^\top \mW \mA)^{-1} a_j. 
\] 
\end{proof}
\begin{lemma}[Number of Coordinate Steps]
\label{lem_coordinate_step_count}For any $0\leq\teps\leq1$, over all $\totaliters$ iterations of Algorithm~\ref{alg_ourAlg}, there
are at most $O(m \max(\alpha, \alpha^{-1}) \log(m/\teps))$   coordinate
steps (see Algorithm~\ref{alg_roundsequential}).
\end{lemma} 
\begin{proof}
There are at most $m$ coordinate steps in each call to Algorithm~\ref{alg_roundsequential}. Combining this with the value of $\totaliters$ in Algorithm~\ref{alg_ourAlg} gives the count of $O(m \alpha^{-1} \log(m/\teps))$ coordinate steps.
\end{proof}

\subsection{Proof of Main Theorem (Sequential)} We now combine the preceding results to prove the main theorem about the sequential algorithm (Algorithm~\ref{alg_ourAlg} with Algorithm~\ref{alg_roundsequential}).
\begin{proof}(Proof of Theorem~\ref{thm_main_thm_seq}) The proof of correctness is the same as that for Theorem~\ref{thm_main_thm} since the parallel and sequential algorithms share the same meta-algorithm. Computing leverage scores in the sequential version (Algorithm~\ref{alg_ourAlg} with Algorithm~\ref{alg_roundsequential}) takes  $O(m\max(\alpha, \alpha^{-1})\log(m/(\alpha \eps)))$ coordinate steps. The costliest component of a coordinate step is computing $a_i^\top (\mA^\top (\mW + \delta_i w_i e_i e_i^\top) \mA)^{-1} a_i$. By the Sherman-Morrison-Woodbury formula, computing this costs $O(n^2)$ for each coordinate. Since the initial cost  to compute $(\mA^\top \mW \mA)^{-1}$ is  $O(mn^2)$, the total run time is $O(\max(\alpha, \alpha^{-1}) mn^2 \log (m / \eps))$. When translated in terms of $p$, this gives $O(p mn^2 \log (mp/ \eps))$ for $p \geq 4$ and $O(p^{-1} mn^2 \log(mp/ \eps))$ for $p \in (2,4)$. 
\end{proof}

\section{A ``One-Step'' Parallel Algorithm}\label{sec:onestepparallelalgo}
We conclude our paper with an alternative algorithm (Algorithm~\ref{alg_oneloop}) in which each iteration avoids any rounding and performs only $\progcode({}\cdot{})$.  

\begin{algorithm}[!ht]
\DontPrintSemicolon
\caption{\textbf{One-Step Algorithm}}\label{alg_oneloop}
\KwIn{Matrix $\mA \in \reals^{m \times n}$, parameter $p > 2$, accuracy $\eps$}
\KwOut{Vector $\hw\in \reals^m_{>0}$ that satisfies  \eqref{def-lewiswts-obj-approxLW}}

For all $i \in [m]$, initialize $w_i^{(0)} = 1$. Set $\alpha = \frac{2}{p-2}$. Set $\teps =  \frac{\alpha^4\eps^4}{(2m (\sqrt{n}+\alpha) (\alpha + \alpha^{-1}))^4}$. \;

Set $\beta = \frac{1}{1000}\min(\alpha^2, 1)$ and $\totaliters = \twopartdef{\Ord(\alpha^{-3}\log (mp/\teps))}{\alpha \in (0, 1]}{\Ord(\alpha^2 \log(mp/\teps))}{\alpha > 1}$
\; 

\For{$k = 0, 1, 2, 3, \dotsc, \totaliters-1$}{

    Let $\eta^{(k)} \in \R^m$ where for all $i \in [m]$ we let $\eta^{(k)}_i = \begin{cases}
    \tfrac{1}{3 \alphamax} & \text{if } \rho_i(w^{(k)}) \geq 1 \\
    \tfrac{1}{3 \alphamax} \beta & \text{if } \rho_i(w^{(k)}) < 1 
    \end{cases}$
 
    $w^{(k+1)} \leftarrow \progcode({w}^{(k)}, [m], \eta^{(k)} )$ 
    \Comment{See  \eqref{eq-wratiostep} and Lemma  \ref{lem_fast_steps_count}} \;
}	

Return $\hw \in \R^m_{>0}$, where $\hw_i = (a_i^\top (\mA^\top \mW^{(\totaliters)} \mA)^{-1} a_i)^{1/\alpha}$. \Comment{See Section~\ref{sec:opti-to-lw}} \; 

\end{algorithm}

\begin{restatable}[Main Theorem (One-Step Parallel Algorithm)]{theorem}{THMonestep}\label{mainthm-onestep} Given a full rank matrix $\mA \in \R^{m\times n}$ and $p\geq 4$, we can compute $\eps$-approximate  Lewis weights \eqref{def-lewiswts-obj-approxLW} in 
 $O( p^3 \log (mp/\eps)$ iterations.  Each iteration computes  the leverage score of one row of $\mD \mA$ for some diagonal matrix $\mD$. The total runtime is $O(p^3 m n^2 \log(m p/\eps))$. 
\end{restatable} 
 We first spell out the key idea of the proof of Theorem~\ref{mainthm-onestep}  in  Lemma~\ref{lem:parallel_iteration_onestep} next, which is that \eqref{eq_grad_cond_for_prog} is maintained in every iteration through the use of varying step sizes, \textit{without explicitly invoking rounding procedures}. Since \eqref{eq_grad_cond_for_prog} always holds, we may use Lemma~\ref{lem_subopt} in bounding the iteration complexity. 
\begin{lemma}[Rounding Condition Invariance] \label{lem:parallel_iteration_onestep}
For any iteration $k \in [\totaliters-2]$ in Algorithm~\ref{alg_oneloop}, if  $\rho_{\max}(w^{(k)}) \leq 1 + \alpha$, then $\rho_{\max}(w^{(k + 1)}) \leq 1 + \alpha$.
\end{lemma}
\begin{proof} 
 By the definition of $\progcode({}\cdot{})$ in \eqref{eq-wratiostep} and choice of $\eta^{(k)}_i$ in Algorithm~\ref{alg_oneloop}, we have, 
\begin{align*}
w^{(k +1)}_{i} &= w_i^{(k)} \cdot\left[1 + \eta_i^{(k)} \left( \frac{\rho_{i}(w^{(k)})-1}{\rho_{i}(w^{(k)})+1}\right) \right] \numberthis\label{eq:eq_w_upper_single_loop}  \\
&\geq w_i^{(k)} (1 - \eta_i^{(k)}) \geq w_i^{(k)}\left( 1 -  \frac{\beta}{3 \alphamax}\right). \numberthis\label{eq:w_upper_single_loop}
\end{align*} 
Applying this inequality to the definition of $\rho(w)$ in \eqref{def-rho-notation},  for all $i \in [m]$, we have
\[ \rho_i(w^{(k +1)}) = \left[\frac{w^{(k +1)}_i}{w^{(k)}_i} \right]^{-\alpha} \frac{1}{[w^{(k)}_i]^{\alpha}}
a_i^\top (\mA^\top \mW^{(k+1)} \mA)^{-1} a_i \leq \left( 1 - \frac{\beta}{3 \alphamax} \right)^{-1} \left[\frac{w^{(k +1)}_i}{w^{(k)}_i} \right]^{-\alpha} \rho_i(w^{(k)}).
\numberthis\label{eq:new_rho_single_loop} \] 
Plugging \eqref{eq:w_upper_single_loop} into \eqref{eq:new_rho_single_loop} when $\rho_i(w^{(k)}) \leq 1$ and using the upper bound on $\beta$ 
yields that 
\[
\rho_i(w^{(k+1)}) \leq  \left( 1 - \frac{\beta}{3 \alphamax} \right)^{-(1 + \alpha)}
\leq 1 + \alpha ~.
\]
If $\rho_i(w^{(k)}) \geq 1$, then \eqref{eq:new_rho_single_loop}, the equality in \eqref{eq:w_upper_single_loop}, the bound on $\beta$, and $\rho_i(w^{(k)}) \leq 1 + \alpha$ imply that 
\begin{align*}
\rho_i(w^{(k+1)}) \leq 
\left(1 -\frac{\beta}{3\alphamax}\right)^{-1}
\left[
1 + \frac{1}{3\alphamax}
\left( 
\frac{\rho_{i}(w^{(k)})-1}{\rho_{i}(w^{(k)})+1}
\right)
\right]^{-\alpha}
\rho_i(w^{(k)})
\leq 1 + \alpha.
\end{align*}
\end{proof} 

\begin{proof}[Proof of Theorem~\ref{mainthm-onestep}]
By our choice of $w_i^{(0)}=1$ for all $i\in [m]$, we have that $\rho_{i}(w^{(0)}) = \sigma_i(w^{(0)}) \leq 1$ by Fact~\ref{fact_projmatrices}. Then applying  Lemma~\ref{lem:parallel_iteration_onestep} yields by induction that $\rho_{\max}(w^{(k)})\leq 1+\alpha$ at every iteration $k$. We may now therefore  upper bound the objective sub-optimality from Lemma \ref{lem_subopt}; as before, combining this with the lower bound on progress from Lemma \ref{lem:parallel_progress} (noticing that $\eta_i \geq \frac{\beta}{3\alphamax}$) yields
\begin{align*}
\fobj(w^{+})-\fobj(\wopt) & \leq \fobj(w)-\fobj(\wopt)-\frac{\beta}{6 \alphamax}\sum_{i=1}^{m}w_i^{1+\alpha} \frac{(\rho_{i}(w)-1 )^{2}}{\rho_{i}(w)+1 }\\
 & \leq\left( 1-\frac{\beta}{30\max(1, \alpha) \max(1, \alpha^{-1})} \right)(\fobj(w)-\fobj(\wopt)). \numberthis\label{fastitercount-11}
\end{align*} 
Thus, $\progcode({}\cdot{})$  decreases $\fobj$. Using $\fobj(w)-\fobj(\wopt)\leq n \log (m/n) \leq m^2$ from Lemma \ref{lem_initialError} and setting \eqref{fastitercount-11} to $\teps$  gives an iteration complexity of $\Ord(\beta^{-1}\alpha^{-1} \log(m/\teps)) = \Ord(\alpha^{-3} \log(m/\teps))$  if $\alpha \in (0, 1]$ and $\Ord(\alpha \beta^{-1} \log (m /\teps)) = \Ord(\alpha  \log (m /\teps))$ otherwise. As in the proofs of Theorems~\ref{thm_main_thm} and ~\ref{thm_main_thm_seq}, we can then invoke Lemma~\ref{lem-fn-to-LW} to construct the vector that is $\eps$-approximate to the Lewis weights. 
\end{proof}

\section{Acknowledgements}
We are grateful to the anonymous reviewers of SODA 2022 for their careful reading and thoughtful comments that helped  us improve our exposition. Maryam Fazel was supported in part by grants NSF TRIPODS II DMS 2023166, NSF TRIPODS CCF 1740551, and NSF CCF 2007036. Yin Tat Lee was supported in part by NSF awards CCF-1749609, DMS-1839116, DMS-2023166, CCF-2105772, a Microsoft Research Faculty Fellowship, Sloan Research Fellowship, and Packard Fellowship. Swati Padmanabhan was supported in part by NSF TRIPODS II DMS 2023166. Aaron Sidford was supported in part by a Microsoft Research Faculty Fellowship, NSF CAREER Award CCF-1844855, NSF Grant CCF-1955039, a PayPal research award, and a Sloan Research Fellowship.   

%% file: appendices.tex
\begin{appendices}
We start with a piece of notation we frequently use in the appendix. For a given vector $x\in \R^m$, we use $\Diag(x)$ to describe the diagonal matrix with $x$ on its diagonal. For a matrix $\mathbf{X}$, we use $\vDiag(\mathbf{X})$ to denote the vector made up of the diagonal entries of $\mathbf{X}$.  Further, recall as stated in Section~\ref{sec:notation}, that given any vector $x$, we use its uppercase boldface name $\mathbf{X} \defeq \Diag(x)$. 
\section{Technical Proofs: Gradient, Hessian, Initial Error, Minimum Progress}\label{sec:technicalproofsapp}
\LEMgradAndHess*
\begin{proof}
The proof essentially follows by combining Lemmas 48 and 49 of \cite{DBLP:journals/corr/abs-1910-08033}. For completeness, we provide the full proof here.  Applying chain rule to the $\log \det$ function and then the definition of $\rho(w)$ from \eqref{def-rho-notation} gives the claim that \[\nabla_{i}\fobj(w)=-(\mA (\mA^{\top}\mW \mA)^{-1}\mA^{\top})_{ii}+w_{i}^{\alpha}=-a_{i}^{\top}(\mA^{\top}\mW \mA)^{-1}a_{i}+w_{i}^{\alpha}=\frac{-\sigma_{i}(w)}{w_{i}}+w_{i}^{\alpha}.\] We now set some notation to compute the Hessian:  let $\mM\defeq \mA(\mA^{\top}\mW \mA)^{-1}\mA^{\top}$, let $h\in \R^m$ be any arbitrary vector, and let $\mH\defeq \Diag(h)$. For $f:\R^n \rightarrow \R$ and for $x, h\in \R^n$ we let $\dir_x f(x)[h]$ denote the directional derivative of $f$ at $x$ in the direction $h$, i.e., $\dir_x f(x)[h] = \lim_{t\rightarrow 0} (f(x+th) - f(x))/t$.  Then we have, 
\begin{align*}
	\dir_w\inprod{h}{-\Diag(\mA(\mA^{\top}\mW\mA)^{-1}\mA^{\top})}[h]	&=\inprod {h}{-\Diag(\mA \dir_w (\mA^{\top}\mW \mA)^{-1}[h]\mA^{\top})}\\
	&= \inprod{h}{\Diag(\mA (\mA^\top \mW \mA)^{-1} \dir_w (\mA^\top \mW \mA)[h] \mA^\top \mW \mA)^{-1} \mA^\top  )} \\ 
	&=\inprod {h}{\Diag(\mM\mH\mM)}\\
	&=\sum_{i,j}h_{i}h_{j}\mM_{ij}\mM_{ji}=\sum_{i,j}h_{i}h_{j}\mM_{ij}^{2},
\end{align*} where the last step follows by symmetry of $\mM$. This implies \begin{align*} \nabla_{ij}^{2}\fobj(w)&=\twopartdef{(a_{i}^{\top}(\mA^{\top}\mW \mA)^{-1}a_{j})^{2}}{i\neq j}{(a_{i}^{\top}(\mA^{\top}\mW \mA)^{-1}a_{j})^{2}+\alpha w_{i}^{\alpha-1}}{otherwise},\end{align*} which, in shorthand, is 
$\nabla^{2}\fobj(w)= \mM\circ \mM+\alpha \mW^{\alpha-1}$. We may express this Hessian as in the statement of the lemma by writing $\mM$ in terms of $\mP(w)$. 
\end{proof}
\LEMinitialError*
\begin{proof} 
We study the two terms constituting the objective in \eqref{def_objective}. First, by choice of $w^{(0)} = \tfrac{n}{m}\1$, we have \[-\log\det  (\mA^\top \mathbf{W}^{(0)} \mA) = -\log \det ((n/m) \mA^\top \mA).\numberthis\label{app-initerr-1}\] Next, since leverage scores always lie between zero and one, the optimality condition for \eqref{def_objective}, $\sigma(\wopt) = ({\wopt})^{1+\alpha}$, implies  $\wopt \leq 1$, which in turn gives $\mwopt \preceq I$. This implies $\mA^\top \mwopt \mA \preceq \mA^\top \mA$. Therefore,  \[ -\log \det (\mA^\top \mA) \leq -\log \det (\mA^\top \mwopt \mA).\numberthis\label{app-initerr-2}\] Combining \eqref{app-initerr-1} and \eqref{app-initerr-2} gives \[-\log \det(\mA^\top \mathbf{W}^{(0)} \mA) \leq -\log \det (\mA^\top \mwopt \mA) + n \log (m/n).\numberthis\label{logdet_initError}\]  Next, observe that $\1^\top (w^{(0)})^{1+\alpha} = m \cdot (n/m)^{1+\alpha}$, and $\1^\top ({\wopt})^{1+\alpha} = \sum_{i=1}^m \sigma_i(\wopt) = n$, where we invoked Fact~\ref{fact_projmatrices}. By now applying $m \geq n$, we get \[ \1^\top (w^{(0)})^{1+\alpha} \leq \1^\top ({\wopt})^{1+\alpha}.\numberthis\label{regularizer_initError}\] Combining \eqref{logdet_initError}, \eqref{regularizer_initError}, and the definition of the objective \eqref{def_objective} finishes the claim. 
\end{proof}
\LEMminfndecrease*
\begin{proof}\label{proof:lemfndecrease}
By the remainder form of Taylor's theorem, for some $t \in [0,1]$ and $\wmid = t w + (1-t) w^+$ 
\begin{equation}
\label{eq:upper_taylor1}
\fobj(w^+) = \fobj(w) + \inprod{\nabla \fobj(w)}{w^+-w} + \frac{1}{2}(w^+-w)^\top \nabla^2 \fobj(\wmid) (w^+ - w). 
\end{equation}
We prove the result by bounding the quadratic form of $\nabla^2 \fobj(\wmid)$ from above and leveraging the structure of $w^+$ and $\nabla \fobj(w)$. Lemma~\ref{lem_gradAndHess} and Fact~\ref{fact_projmatrices} imply that
\[
\nabla^2 \fobj(\wmid) = \mwmid^{-1} \mP(\wmid)^{(2)} \mwmid^{-1}  +\alpha \mwmid^{\alpha-1} 
\preceq \mwmid^{-1} \Sigma(\wmid) \mwmid^{-1} + \alpha \mwmid^{\alpha-1} 
~. \numberthis\label{eq:hess-temp-upperbound}
\]
Further, the positivity of $w_i$ and $\sigma_i(w)$ and the non-negativity of $\eta$ and $\rho$ imply that $(1- \norm{\eta}_\infty) w_i \leq w_i^+ \leq (1 + \norm{\eta}_\infty) w_i$  for all  $i \in [m]$. Since $\norm{\eta}_\infty \leq \frac{1}{3\alphamax}$, this implies that 
$$
(1- \tfrac{1}{3\alphamax}) w_i \leq \wmid_i \leq (1+ \tfrac{1}{3\alphamax}) w_i ~ \text{ for all } ~ i \in [m] ~.
$$
Consequently, for all $i \in [m]$, we bound the first term of \eqref{eq:hess-temp-upperbound} as 
\begin{align*}
\left[ \mwmid^{-1} \Sigma(\wmid) \mwmid^{-1} \right]_{ii}
&=
e_i^\top \mwmid^{-1/2} \mA (\mA^\top \mwmid \mA)^{-1} \mA^\top \mwmid^{-1/2} e_i 
    = \frac{1}{\wmid_i} a_i^\top (\mA^\top \mwmid \mA)^{-1} a_i \\ 
&\leq (1-\tfrac{1}{3\alphamax})^{-1} \frac{1}{w_i} a_i^\top (\mA^\top \mwmid \mA)^{-1} a_i 
    \leq (1-\tfrac{1}{3\alphamax})^{-2} \frac{1}{w_i} a_i^\top (\mA^\top \mW \mA)^{-1} a_i \\
    &= (1-\tfrac{1}{3\alphamax})^{-2} \left[ \mW^{-1} \Sigma(w) \mW^{-1} \right]_{ii} 
    \preceq 3 \left[ \mW^{-1} \Sigma(w) \mW^{-1} \right]_{ii} \numberthis\label{app-hessbound-1}
\end{align*} 
Further, when $\alpha \in (0, 1]$, we bound the second term of \eqref{eq:hess-temp-upperbound} as \[\mwmid^{\alpha - 1} 
\preceq 
(1 - \tfrac{1}{3\alphamax})^{\alpha - 1} \mW^{\alpha - 1} \preceq (1 - \tfrac{1}{3\alphamax})^{- 1} \mW^{\alpha - 1}
\preceq 3 \mW^{\alpha - 1},\numberthis\label{app-hessbound-2}\]  
and when $\alpha \geq 1$, we have \[\mwmid^{\alpha - 1} 
\preceq (1 + \tfrac{1}{3\alphamax})^{\alpha - 1} \mW^{\alpha - 1}
\preceq 
\exp(\frac{\alpha - 1}{3\alphamax}) \mW^{\alpha - 1}
= \exp(\frac{\alpha - 1}{3\alpha}) \mW^{\alpha - 1}
\preceq 3 \mW^{\alpha - 1}.\numberthis\label{app-hessbound-3}\] Using \eqref{app-hessbound-1}, \eqref{app-hessbound-2}, and \eqref{app-hessbound-3} in \eqref{eq:hess-temp-upperbound}, we have that in all cases 
\begin{equation*}
\nabla^2 \fobj(\wmid)
\preceq 3\left[ 
\mW^{-1} \Sigma(w) \mW^{-1} + \alpha \mW^{\alpha-1} 
\right]
\preceq 3\alphamax \mW^{-1} \left[ 
 \Sigma(w) + \mW^{1 + \alpha} 
\right]
\mW^{-1} 
~.
\end{equation*}
Applying to the above Loewner inequality the definition of $w^+$ gives
\begin{align}
(w^+-w)^\top \nabla^2 \fobj(\wmid) (w^+ - w)
&\leq 
\sum_{i \in [m]}
	3\alphamax
	\cdot (w_i^{1 + \alpha} + \sigma_i(w)) \cdot \left(\eta_i \cdot \frac{\rho_i(w)-1}{\rho_i(w) + 1} \right)^2 \nonumber \\
&=
\sum_{i \in [m]}
	3\alphamax \cdot \eta_i^2
	\cdot w_{i}^{1 + \alpha} 
	\cdot \frac{(\rho_i(w) - 1)^2}{\rho_i(w) + 1} ~.
	\label{eq:upper_hess}
\end{align} Next, recall that by Lemma~\ref{lem_gradAndHess}, $\left[ \nabla \fobj(w)\right]_i = w_i^{-1} \cdot (w_i^{1+\alpha} - \sigma_i(w))$ for all $i \in [m]$. Consequently, 
\[\inprod{\nabla \fobj(w)}{w^+-w} = 
\sum_{i \in [m]} 
	(w_i^{1 + \alpha} - \sigma_i(w)) 
	\cdot
	\left(\eta_i \cdot \frac{\rho_i(w)-1}{\rho_i(w) + 1} \right)
=
-
\sum_{i \in [m]}
	\eta_i
	\cdot w_{i}^{1 + \alpha} 
	\cdot \frac{(\rho_i(w) - 1)^2}{\rho_i(w) + 1} ~.
	\numberthis \label{eq:upper_grade}
	\]
Combining \eqref{eq:upper_taylor1}, \eqref{eq:upper_hess}, and \eqref{eq:upper_grade} yields that
\begin{align*}
\fobj(w^+) 
&\leq  \fobj(w) + \sum_{i \in [m]} 
	\left(- \eta_i + \frac{3\alphamax \eta_i^2}{2}\right) \cdot w_i^{1 + \alpha} \cdot \frac{(\rho_i(w) - 1)^2}{\rho_i(w) + 1} ~.
\end{align*} 
The result follows by plugging in $\eta_i \in [0, (3\alphamax)^{-1}]$, as assumed. 
\end{proof}

\section{From Optimization Problem to Lewis Weights }\label{sec:opti-to-lw}
The goal of this section is to prove how to obtain $\eps$-approximate Lewis weights from an $\teps$-approximate solution to the problem in \eqref{def_objective}. Our proof strategy is to first utilize the fact that the vector $\wround$ obtained after the rounding step following the \texttt{for} loop of Algorithm~\ref{alg_ourAlg}  satisfies the properties of being $\teps$-suboptimal (additively) and also the rounding condition \eqref{eq_grad_cond_for_prog}. In Lemma~\ref{lem-fn-to-LW}, the $\teps$-suboptimality is used to show a bound on $\|\sigma(\wround) - \wround^{1+\alpha}\|_{\infty}$. Coupled with the rounding condition, we then show in Lemma~\ref{lem_rounding} that $\widehat{\wround}$ constructed as per the last line of Algorithm~\ref{alg_ourAlg} then satisfies approximate optimality, $\sigma(\widehat{w}) \approx_{\delta} \widehat{w}^{1+\alpha}$,  for some small $\delta>0$. In Lemma~\ref{lem:lastlemma}, we finally relate this approximate optimality to coordinate-wise multiplicative closeness between $\widehat{w}$ and the vector of true Lewis weights. Finally, in Lemma~\ref{lem-fn-to-LW}, we pick the appropriate approximation factors for each of the lemmas invoked and prove the desired approximation. Since the vector $w^{\totaliters}$ obtained at the end of the \texttt{for} loop of Algorithm~\ref{alg_oneloop} also satisfies the aforementioned properties of $\wround$, the same set of lemmas apply to Algorithm~\ref{alg_oneloop} as well. We begin with  some technical lemmas. 
\subsection{From Approximate Closeness to Approximate Optimality}
\begin{lemma}\label{lem_rounding}
Let $w \in \R^m_{>0}$ such that  $\|\sigma(w) - w^{1+\alpha}\|_\infty \leq \beps$ for some parameter $0<\beps \leq \frac{1}{100m^2 (\alpha + \alpha^{-1})^2}$ and also let $\rho_{\max}(w) \leq 1+\alpha$. Define $\hw_i = (a_i^\top (\mA^\top \mW \mA)^{-1} a_i)^{1/\alpha}$. Then, for the parameter $\delta = 20 \sqrt{\beps} m (\alpha + \alpha^{-1})$, we have that $\sigma(\hw) {\approx_\delta} \hw^{1+\alpha}$. 
\end{lemma}
\begin{proof} Our strategy to prove $\sigma(\hw) \approx_\delta \hw^{1+\alpha}$ involves first noting that this is the same as proving $\hw^{-1} \cdot \sigma(\hw) \approx_\delta \hw^\alpha$ and, from the definition of $\hw$ in the statement of the lemma, to instead prove $\mA^\top \mWW \mA \approx_\delta \mA^\top \mW \mA$. 

To this end, we split $\mW$ into two matrices based on the size of its coordinates, setting the following notation. Define $\mWl$ to be the diagonal matrix $\mW$ with zeroes at indices corresponding to $w> \eta$, and $\mWWl$ to be the diagonal matrix $\mWW$ with zeroes at indices corresponding to $w> \eta$. We first show that $\mA^\top \mWWl \mA$ and $\mA^\top \mWl \mA$ are small compared to $\mA^\top \mW \mA$ and can therefore be ignored in the preceding desired approximation. We then prove that for $w>\eta$, we have $w \approx_\delta \hw$. This proof technique is inspired by Lemma 4 of \cite{vaidya1989new}.  

First, we prove that $\mA^\top \mWWl \mA$ is small as compared to $\mA^\top \mWg \mA$. Since \eqref{eq_grad_cond_for_prog} is satisfied, it means   $$a_i^\top (\mA^\top \mW \mA)^{-1} a_i =\sigma_i(w) \cdot w_i^{-1} \leq (1+\alpha) w_i^\alpha.$$ Combining this with the definition of $\hw_i$ as in the statement of the lemma, we may use non-negativity of $\alpha$ to derive \[\hw_i \leq (1+\alpha)^{1/\alpha} w_i \leq 3w_i. \numberthis\label{eq_wbar_bound}\] We apply this inequality in the following expression to obtain 
\begin{align*}
    \Tr ((\mA^\top \mWWl \mA) (\mA^\top \mW \mA)^{-1}) &= \sum_{w_i \leq \eta} \hw_i (a_i^\top (\mA^\top \mW \mA)^{-1} a_i)\\
                    &= \sum_{w_i \leq \eta} (a_i^\top (\mA^\top \mW \mA)^{-1} a_i)^{1+ 1/\alpha} \\
                    &\leq (1+\alpha)^{1 + 1/\alpha} \sum_{w_i \leq \eta} w_i^{1+\alpha} \\
                    &\leq 3(1+\alpha) m \eta^{1+\alpha}. \numberthis\label{eq_atwba_small}
\end{align*} This implies that\footnote{Given $X, Y\succeq 0$, we have $Y^{1/2} X Y^{1/2} \succeq 0$. Then, if $\Tr(XY) \leq 1$, we have $\Tr(Y^{1/2} X Y^{1/2}) \leq 1$, and combining these with the previous matrix inequality, we  conclude that $Y^{1/2} X Y^{1/2} \preceq I$, which implies that $X\preceq Y^{-1}$. } \[ \mA^\top \mWWl \mA \preceq 3(1+\alpha)m \eta^{1+\alpha} \mA^\top \mW \mA. \numberthis\label{eq_eq_atwba_bound}\]

Our next goal is to bound $\mA^\top \mWWg \mA$ in terms of $\mA^\top \mW \mA$, which we do by first bounding it in terms of $\mA^\top \mWg \mA$ and then bounding $\mA^\top \mWg \mA$ in terms of $\mA^\top \mW \mA$. By definition, $\hw_i^\alpha = \sigma_i(w) \cdot w_i^{-1}$. Further, by assumption, $\|\sigma(w) - w^{1+\alpha}\|_{\infty} \leq \beps$. Therefore, for any $w_i \geq \eta$  \[\hw_i^\alpha \leq (w_i^{1+\alpha} + \beps) \cdot w_i^{-1} \leq  (1+\beps/\eta^{1+\alpha})w_i^{1+\alpha} \cdot w_i^{-1} = (1+\beps/\eta^{1+\alpha}) w_i^\alpha, \] and \[\hw_i^\alpha \geq (w_i^{1+\alpha} - \beps) \cdot w_i^{-1} \geq  (1-\beps/\eta^{1+\alpha})w_i^{1+\alpha} \cdot w_i^{-1} = (1-\beps/\eta^{1+\alpha})w_i^\alpha.\] By our choice of $\beps$, for $w_i \geq \eta$, we have\[\left( 1- \frac{2\beps}{\alpha \eta^{1+\alpha}} \right) w_i \leq \hw_i \leq \left( 1+ \frac{2\beps}{\alpha \eta^{1+\alpha}} \right) w_i.\numberthis\label{eq_wbg_bound}\]

Further, we have the following inequality: 
 \[\mA^\top \mWg \mA \preceq \mA^\top \mW \mA. \numberthis\label{eq_atwa_bound}\]  Hence, we can combine (\ref{eq_atwa_bound}), (\ref{eq_wbg_bound}), and (\ref{eq_eq_atwba_bound}) to see that 
\begin{align*}
    \mA^\top \mWW \mA &= \mA^\top \mWWg \mA + \mA^\top \mWWl \mA\\
                    &\preceq \left( 1+ \frac{2\beps}{\alpha \eta^{1+\alpha}}\right) \mA^\top \mWg \mA + 3(1+\alpha) m \eta^{1+\alpha} \mA^\top \mW \mA\\
                    &\preceq \mA^\top \mW \mA \left(   1+ \frac{2\beps}{\alpha \eta^{1+\alpha}} +  3(1+\alpha)m \eta^{1+\alpha} \right). 
\end{align*}
Set $\eta^{1+\alpha} = \sqrt{\beps}$ for the  upper bound. 

For the lower bound, we bound $\mA^\top \mWl \mA$ and, therefore, also $\mA^\top \mWg \mA$. Observe that 
\begin{align*}
    \Tr ((\mA^\top \mWl \mA)(\mA^\top \mW \mA)^{-1}) &= \sum_{w_i \leq \eta}  w_i a_i^\top (\mA^\top \mW \mA)^{-1} a_i = \sum_{w_i \leq \eta} \sigma_i(w) \\ 
                                                  &\leq \sum_{w_i \leq \eta} (w_i^{1+\alpha} + \beps) \leq  m (\eta^{1+\alpha} + \beps), 
\end{align*} where the second step is by $\|\sigma(w) - w^{1+\alpha} \|_\infty \leq \beps$, as assumed in the lemma. This implies that\[ \mA^\top \mWl \mA \preceq m (\eta^{1+\alpha} + \beps) \mA^\top \mW \mA,\] and therefore that \[ \mA^\top \mWg \mA \succeq (1-m(\eta^{1+\alpha} + \beps)) \mA^\top \mW \mA.\]  Repeating the method for the upper bound then finishes the proof. 
\end{proof}
\subsection{From Approximate Optimality to Approximate Lewis Weights}
In this section, we go from the previous notion of approximation to the one we finally seek in \eqref{def-lewiswts-obj-approxLW}. Specifically,  we show that if $\sigma(w)\approx_{\beta}w^{1+\alpha}$, then
$w\approx_{O((\beta/\alpha)\sqrt{n})}\lw$. To prove this,
we first give a technical result. We recall notation stated in Section~\ref{sec:notation}: for any projection matrix $\mP(w) \in \R^{m \times m}$, we have the vector of leverage scores $\sigma(w) = \vDiag(\mP(w))$.
\begin{claim}
\label{lem:proj_fact} For any projection matrix $\mP(w)\in\R^{m\times m}$,  $\alpha\geq0$,
and vector $x\in\R^{m}$, we have that 
\[
\norm{\left[\mP(w)^{(2)}+\alpha\mSigma(w)\right]^{-1}\mSigma(w) x}_{\infty}
\leq\frac{1}{\alpha}\norm{x}_{\infty}+\frac{1}{\alpha^2}\norm{ x}_{\mSigma(w)}
\leq\left(\frac{1+\sqrt{n}/\alpha}{\alpha}\right)\norm{x}_{\infty}
\]
\end{claim}

\begin{proof}
 
Let $y\defeq\left[\mP(w)^{(2)}+\alpha\mSigma(w)\right]^{-1}\mSigma(w) x$.
Since $\mZero\preceq\mP(w)^{(2)}\preceq\mSigma(w)$ (Fact~\ref{fact_projmatrices}), we have that $\mSigma(w)\preceq\frac{1}{\alpha}\left[\mP(w)^{(2)}+\alpha\mSigma(w)\right]$
and $(\mP(w)^{(2)}+\alpha\mSigma(w))^{-1}\preceq\alpha^{-1}\mSigma(w)^{-1}$.
Consequently, taking norms in terms of these matrices gives 
\begin{align}
\norm{y}_{\mSigma(w)}=\norm{\left[\mP(w)^{(2)}+\alpha\mSigma(w)\right]^{-1}\mSigma(w) x}_{\mSigma(w)} & \leq\frac{1}{\sqrt{\alpha}}\norm{\mSigma(w) x}_{\left[\mP(w)^{(2)}+\alpha\mSigma(w)\right]^{-1}}\leq\frac{1}{\alpha}\norm{x}_{\mSigma(w)}\,.\label{eq:proj_bound_1}
\end{align}
Next, since by Lemma 47 of \cite{DBLP:journals/corr/abs-1910-08033}, $\norm{\mSigma(w)^{-1}\mP(w)^{(2)}z}_{\infty}\leq\norm{z}_{\mSigma(w)}$
  for all $z\in\R^{m}$, we see that $\left|[\mP(w)^{(2)}y]_{i}\right|\leq\sigma_{i}(w)\norm{y}_{\mSigma(w)}$ for all $i \in [m]$, 
and since by definition of $y$, we have $[(\mP(w)^{(2)}+\alpha\mSigma(w))y]_{i}=\sigma_{i}(w) x_{i}$ for
all $i\in[m]$, we have that 
\begin{equation}
\norm{y}_{\infty}
=\max_{i\in[m]}|y_{i}|
=\max_{i\in[m]}\left|\frac{1}{\alpha}x_{i}+\frac{1}{\alpha\sigma_{i}(w)}\left[\mP(w)^{(2)}y\right]_{i}\right|\leq\frac{1}{\alpha}\norm{x}_{\infty}+\frac{1}{\alpha}\norm{y}_{\mSigma(w)}\,.\label{eq:proj_bound_2}
\end{equation}
Combining \eqref{eq:proj_bound_1} and \eqref{eq:proj_bound_2} and using
that $\sum_{i\in[m]}\sigma_{i}(w) \leq n$ yields the claim. 
\end{proof}

\begin{lemma}\label{lem:lastlemma}
Let 
$\widehat{w}\in\R_{>0}^{m}$ be a vector that satisfies approximate optimality of \eqref{def_objective} in the following sense: $$\sigma(\widehat{w})=\widehat{\mW}^{1+\alpha}v, \text{ 
for }\exp(-\mu)\ones\leq v\leq\exp(\mu)\ones.$$ Then, $\widehat{w}$ is also coordinate-wise multiplicatively close to $\lw$, the true vector of Lewis weights, as formalized below.
\[
\exp\left(-\frac{1}{\alpha}(1+\sqrt{n}/\alpha)\mu\right)\lw\leq \widehat{w}\leq\exp\left(\frac{1}{\alpha}(1+\sqrt{n}/\alpha)\mu\right)\lw\,.
\]
\end{lemma}
\begin{proof}
For all $t\in[0,1]$, let $[v_{t}]_i=[v_i^t]$ so that $v_{1}=v$
and $v_{0}=\ones$. Further, for all $t\in[0,1]$, let $w_{t}$ be
the unique solution to 
\[
w_{t}=\argmin_{w\in\R_{> 0}^{m}}f_{t}(w)\defeq-\log\det\left(\mA^{\top}\mW\mA\right)+\frac{1}{1+\alpha}\sum_{i\in[m]}[v_{t}]_{i}w_{i}^{1+\alpha}.\numberthis\label{def-wt-app}
\]
Then we have the following gradients. 
\begin{align*}
\nabla_{w}f_{t}(w) & =-\mW^{-1}\sigma(w)+\mW^{\alpha}v_{t}\,,\\
\nabla_{w} (\frac{d}{dt}f_{t})(w) & =\mW^{\alpha}\frac{d}{dt}v_{t}=\mW^{\alpha}v_{t}\ln(v)\,\numberthis\label{app-eq-grads-hess}\\
\nabla^2_{ww}f_{t}(w) &=\mW^{-1}\left[\mP(w)^{(2)} +\alpha\mW^{1+\alpha}\mV\right]\mW^{-1}\,\numberthis\label{app-eq-grads-hess-2}.
\end{align*}
Consequently, by optimality of $w_t$ as defined in \eqref{def-wt-app}, we have $
\zeros=\nabla_{w}f_{t}(w_t)=-\mW_{t}^{-1}\sigma(w_t)+\mW_{t}^{\alpha}v_{t}.$
Rearranging the terms of this equation yields that \[\sigma(w_t)=\mW_{t}^{1+\alpha}v_{t},\numberthis\label{app-sigmawt-eq}\] 
and therefore $w_{1}=\widehat{w}$ and $w_{0}=\lw$. To prove the
lemma, it therefore suffices to bound
\begin{equation}
\ln(\widehat{w}/\lw)=\ln(w_1/w_0)=\int_{t=0}^{1}\left[\frac{d}{dt}\ln(w_{t})\right]dt=\int_{t=0}^{1}\mW_{t}^{-1}\left[\frac{d}{dt}w_{t}\right]dt\,.\label{eq:int_form}
\end{equation}To bound \eqref{eq:int_form}, it remains to compute $\frac{d}{dt}w_{t}$
and apply Claim~\ref{lem:proj_fact}. To do this, note that 
\[
\mZero=\frac{d}{dt}\grad_{w}\left[f_t(w_{t})\right]=\nabla_{w}(\frac{d}{dt}f_{t})(w_t)+\nabla^2_{ww}f_{t}(w_t)\cdot\frac{d}{dt}w_{t}\,.
\]
Using that $\mP(w_{t})^{(2)}+\mW_{t}^{1+\alpha}\mV_{t}\succ\mZero$, 
we have, by rearranging the above equation and applying \eqref{app-eq-grads-hess} and \eqref{app-eq-grads-hess-2} that 
\begin{align*}
\frac{d}{dt}w_{t} & =-\left[\nabla^2_{ww}f_t(w_{t})\right]^{-1} \cdot \left[\nabla_{w}(\frac{d}{dt}f_{t})(w_t)\right] =-\mW_{t}\left[\mP(w_t)^{(2)}+\alpha\mW_{t}^{1+\alpha}\mV_{t}\right]^{-1}\mW_{t}^{1+\alpha}v_{t}\ln(v)\,.\numberthis\label{app-der-bound}
\end{align*}
Applying \eqref{app-sigmawt-eq} to \eqref{app-der-bound}, we have that 
\begin{align*}
\mW_{t}^{-1}\left[\frac{d}{dt}w_{t}\right] & =-\left[\mP(w_{t})^{(2)}+\alpha\mSigma(w_t)\right]^{-1}\mSigma(w_{t})\ln(v)\,.
\end{align*}
Applying Claim~\ref{lem:proj_fact} to the above equality, substituting in \eqref{eq:int_form} and $\norm{\ln(v)}_{\infty}\leq\mu$ therefore yields  
\begin{align*} \norm{\ln(\widehat{w}/\lw)}_{\infty}=\norm{\ln(w_1/w_0)}_{\infty}
& \leq\int_{t=0}^{1}\norm{\mW_{t}^{-1}\left[\frac{d}{dt}w_{t}\right]}_{\infty}dt\leq\int_{t=0}^{1}\left(\frac{1+\sqrt{n}/\alpha}{\alpha}\right)\mu dt\,.
\end{align*}
\end{proof}
\subsection{From Optimization Problem to Approximate Lewis Weights}

\LEMroundingToLW*
\begin{proof}
We are given a vector $w\in \R^m$ satisfying $\fobj(\wopt)\leq\fobj(w)\leq\fobj(\wopt)+\teps$.  Then by  Lemma \ref{lem_subopt}, we have that $ \frac{(\sigma_i(w) - {w}_i^{1+\alpha})^2}{\sigma_i(w) + {w}_i^{1+\alpha}} \leq \teps$ for each $i\in [m]$. This bound implies that ${w}_i \leq 3$ for all $i$ because, if not, then because of $\sigma_i(w) \in [0, 1]$ and the decreasing nature of $(x-a)^2/(x+a)$ over $x\in [0, 1]$ for a fixed $a\geq 3$, we obtain  $\frac{(\sigma_i(w) - {w}_i^{1+\alpha})^2}{\sigma_i(w)+{w}_i^{1+\alpha}}\geq \frac{(1-{w}_i^{1+\alpha})^2}{1+{w}_i^{1+\alpha}}\geq 1$, a contradiction. Therefore
$\|\sigma(w) - w^{1+\alpha}\|_\infty \leq 2\sqrt{\teps}$. Coupled with the provided guarantee $\rho_{\max}(w) \leq 1+\alpha$, we see that the requirements of Lemma~\ref{lem_rounding} are met with $\beps = 2\sqrt{\teps}$, for  $\teps \defeq \frac{\heps^4}{(25m (\alpha + \alpha^{-1}))^4}$, and Algorithm \ref{alg_ourAlg} therefore guarantees a $\hw$ satisfying $\sigma(\hw) \approx_{\heps} \hw^{1+\alpha}$. Therefore, we can now apply Lemma~\ref{lem:lastlemma} with $\mu = \heps$, and choosing $\heps = \tfrac{\alpha^2}{\alpha + \sqrt{n}}\eps$ lets us conclude that $\hw_i \approx_{\eps} \lw_i$, as claimed. 
\end{proof}

\section{A Geometric View of Rounding}\label{explain-rounding-name}
At the end of Algorithms \ref{alg_roundparallel} and \ref{alg_roundsequential}, the iterate $w$ satisfies the condition  $\rho_{\max}(w)\leq1+\alpha$. 
We now show the geometry implied by the preceding condition, thereby provide the reason behind the terminology ``rounding.'' 

\begin{lemma}
Given $w\in\mathbb{R}_{>0}^{m}$ such that $\rho_{\max}(w)\leq1+\alpha$. Define the ellipsoid  $\mathcal{E}(w):=\{x:x^{\top}\mA^{\top}\mW \mA x\leq1\}$. Then, we have that
\[
\mathcal{E}(w)\subset\{x \in \R^n ~ | ~ \|\mW^{-\alpha/2}\mA x\|_{\infty}\leq\sqrt{1+\alpha}\}.
\]
\end{lemma}
\begin{proof}
Consider any point $x\in \mathcal{E}(w)$. Then,
 by Cauchy-Schwarz inequality  and $\rho_{\max}(w) \leq 1+ \alpha$,
\begin{align*}
\|\mW^{-\alpha/2}\mA x\|_{\infty} & =\max_{i\in [m] }e_{i}^{\top}\mW^{-\alpha/2}\mA x = \max_{i\in [m] }e_{i}^{\top}\mW^{-\alpha/2}\mA (\mA^{\top}\mW \mA)^{-\frac{1}{2}}(\mA^{\top}\mW \mA)^{\frac{1}{2}}x\\
 & \leq\max_{i\in [m] }\sqrt{e_{i}^{\top}\mW^{-\alpha/2}\mA(\mA^{\top}\mW \mA)^{-1}\mA^{\top}\mW^{-\alpha/2}e_{i}}\sqrt{x^{\top}\mA^{\top}\mW \mA x}\\
 & \leq\max_{i\in [m]}\sqrt{e_{i}^{\top}\mW^{-\alpha/2}\mA(\mA^{\top}\mW\mA)^{-1}\mA^{\top}\mW^{-\alpha/2}e_{i}} = \max_{i\in [m]} \sqrt{\frac{\sigma_{i}(w)}{w_{i}^{1+\alpha}}}\leq\sqrt{1+\alpha}.
\end{align*} 
\end{proof}

\section{Explanations of Runtimes in Prior Work}\label{cohenpeng-determinant}
The convex program \eqref{opt-lewisell} formulated by \cite{CohenPeng2015} has a variable size of $n^2$. Therefore, by \cite{lee2015faster}, the number of iterations to solve it using the cutting plane method is $O(n^2\log (n\eps^{-1})$,  each iteration  computing $a_i^\top \mM a_i$ for $i \in [m]$. This can be computed by multiplying an $n\times n$ matrix with an $n\times m$ matrix, which costs between $O(mn)$ (at least the size of the larger input matrix) and $O(mn^2)$ (each entry of the resulting  $m \times n$ matrix obtained by an inner product of length $n$ vectors). Further, there is  at least a total of $O(n^6)$ additional work done by the cutting plane method. This gives us a cost of at least $n^2 (mn + n^4)$.  The runtime of \cite{YinTatThesis} follows from Theorem $5.3.4$. 

\end{appendices}